\documentclass[envcountsame]{article}
\usepackage[latin1]{inputenc}
\usepackage{amsmath}
\usepackage{amsfonts}
\usepackage{amssymb,amsthm}
\usepackage{enumerate}
\usepackage{graphicx}
\usepackage{xcolor}
\usepackage{array}
\definecolor{violet}{rgb}{0.5,0,0.5}
\usepackage{mathrsfs}
\usepackage{bbm}
\usepackage{fourier-orns}
\usepackage{arydshln}
\usepackage{pgf,tikz,pgflibraryarrows}
\usepackage{pgflibrarysnakes,pgflibraryshapes}
\usepackage{xcolor}
\usepackage{todonotes}
\usepackage{stmaryrd}
\usepackage{xspace}
\usepackage{times}
\usepackage{tabularx}

\definecolor{vertCitron}{rgb}{0.6,0.8,0.2}

\definecolor{fuchsia}{rgb}{0.82,0,0.37}

\definecolor{turquoise}{rgb}{0,0.40,0.55}

\definecolor{orange}{rgb}{0.93,0.43,0.06}

\definecolor{ciel}{rgb}{0,100,120}

\newcommand{\leng}[1]{\ensuremath{|#1|}}
\newcommand{\sub}[1]{\ensuremath{Sub\left(#1\right)}}
\newcommand{\sem}[1]{\ensuremath{\left\llbracket#1\right\rrbracket}}

\newcommand{\nbclocks}[1]{\ensuremath{\left\|#1\right\|}}
\newcommand{\configs}[1]{\ensuremath{\mathrm{Config}\left(#1\right)}}
\newcommand{\Aa}{\ensuremath{\mathcal{A}}}

\newcommand{\Ii}{\ensuremath{\mathcal{I}}}
\newcommand{\Bb}{\ensuremath{\mathcal{B}}}
\newcommand{\Cc}{\ensuremath{\mathcal{C}}}
\newcommand{\Tt}{\ensuremath{\mathcal{T}}}
\newcommand{\appf}[1]{\ensuremath{\mathit{APP}_{#1}}}
\newcommand{\timestep}[1]{\ensuremath{\overset{#1}{\rightsquigarrow}}}
\newcommand{\appfunc}[1]{\ensuremath{f_{#1}^\star}}
\newcommand{\mername}{\ensuremath{{\sf Merge}}\xspace}
\newcommand{\mer}[1]{\ensuremath{\mername\left(#1\right)}}
\newcommand{\re}[1]{\ensuremath{{\sf reset}\left(#1\right)}}
\newcommand{\gu}[1]{\ensuremath{{\sf guard}\left(#1\right)}}
\newcommand{\de}[1]{\ensuremath{{\sf loop}\left(#1\right)}}
\newcommand{\loc}{{\sf loc}}
\newtheorem{definition}{Definition}
\newtheorem{theorem}{Definition}
\newtheorem{proposition}{Proposition}
\newtheorem{lemma}{Lemma}
\newtheorem{example}{Example}
\newtheorem{corollary}{Corollary}

\newcommand{\R}{\mathbb{R}}
\newcommand{\N}{\mathbb{N}}

\newcommand{\ATA}{OCATA\xspace}
\title{On MITL and alternating timed automata}

\author{Thomas Brihaye${}^1$ \and Morgane Estiévenart${}^1$
\and Gilles Geeraerts${}^2$}

\date{${}^1$ Universit\'e de Mons, Belgium,\\ ${}^2$ Universit\'e Libre
  de Bruxelles, Belgium}

\begin{document}

\maketitle

\begin{abstract}
  \emph{One clock alternating timed automata} (\ATA) have been
  recently introduced as natural extension of (one clock) timed
  automata to express the semantics of MTL \cite{OW05}. We consider
  the application of \ATA to problem of model-checking MITL formulas
  (a syntactic fragment of MTL) against timed automata. We introduce
  a new semantics for \ATA where, intuitively, clock valuations are
  \emph{intervals} instead of \emph{single values in $\R$}. Thanks to
  this new semantics, we show that we can \emph{bound} the number of
  clock copies that are necessary to allow an \ATA to recognise the
  models of an MITL formula. Equipped with this technique, we propose
  a new algorithm to translate an MITL formula into a timed automaton,
  and we sketch several ideas to define new model checking algorithms
  for MITL.
\end{abstract}

\section{Introduction}
\emph{Automata-based model-checking} \cite{CGP01,VW86} is nowadays a
well-established technique for establishing the correctness of
computer systems. In this framework, the system to analyse is modeled
by means of a \emph{finite automaton} $A$ whose accepted language
consists of all the traces of the system. The \emph{property} to prove
is usually expressed using a \emph{temporal logic formula} $\Phi$,
whose set of models is the language of all \emph{correct
  executions}. For instance, the LTL formula $\Box (p\implies\Diamond
q)$ says that every $p$-event should eventually be followed by a
$q$-event. Then, establishing correctness of the system amounts to
showing that the language $L(A)$ of the automaton is included in the
language $\sem{\Phi}$ of the formula. In practice, automata-based
model checking algorithms first negate the formula and translate
$\neg\Phi$ into an automaton $A_{\neg\Phi}$ that recognises the
\emph{complement} of $\sem{\Phi}$, i.e., the set of all
\emph{erroneous traces}. Then, the algorithm proceeds by computing the
synchronous product $A\times A_{\neg\Phi}$ and check whether
$L(A\times A_{\neg\Phi})=\emptyset$, in which case the system respects
the property.

While those techniques are now routinely used to prove the correctness
of huge systems against complex properties \cite{BCMDH92}, the model
of finite automata and the classical temporal logics such as LTL are
sometimes not \emph{expressive} enough because they can model the
\emph{possible sequences of events}, but cannot express
\emph{quantitative properties about the (real) time elapsing between
  successive events}. To overcome these weaknesses, Alur and Dill
\cite{AD94} have proposed the model of \emph{timed automata}, that
extends finite automata with a finite set of (real valued) clocks. A
real-time extension of LTL is the \emph{Metric Temporal Logic} (MTL)
that has been proposed by Koymans \cite{K90} and consists in labeling
the modalities with time intervals. For instance
$\Box(p\implies\Diamond_{[1,2]}q)$ means `at all time, each $p$ should
be followed by a $q$-event that occurs \emph{between $1$ and $2$ time
  units later}'. Unfortunately, the satisfiability and model-checking
of MTL are undecidable on infinite words \cite{H91}, and non-primitive
recursive on finite words \cite{OW07}.

An interesting alternative is the Metric Interval Temporal Logic
(MITL), that has been proposed by Henzinger \textit{et al.}
\cite{AFH96}. MITL is a syntactic fragment of MTL where singular
intervals are disallowed on the modalities. Thanks to this
restriction, MITL model-checking is \textsc{ExpSpace}-c, even on
infinite words. MITL thus seems a good compromise between
\emph{expressiveness} and \emph{complexity}. In their seminal work,
Henzinger \textit{et al.} provide a construction to translate an MITL
formula $\Phi$ into a timed automaton $\Bb_\Phi$, from which the
automaton-based model checking procedure sketched above can be
applied. Although this procedure is foundational from the theoretical
point of view, it does not seem easily amenable to efficient
implementation: the construction is quite involved, and requests that
$\Bb_\Phi$ be completely built before the synchronous product with the
system's model can be explored. Note that an alternative technique,
based on the notion of \emph{signal} has been proposed by Maler
\textit{et al.}  \cite{MNP06}. However the semantics of MITL assumed
there \emph{slightly} differs from that of \cite{AFH96}, whereas we
stick to the original MITL semantics.

Since MITL is a \emph{syntactic fragment} of MTL, all the techniques
developed by Ouaknine and Worrell \cite{OW05} for MTL can be applied
to MITL. Their technique relies on the notion of \emph{alternating
  timed automaton with one clock} (\ATA), an extension of timed
automata. Intuitively an \ATA can create several copies of itself that
run in parallel and must \emph{all} accept the suffix of the word.
For example, Fig.~\ref{ExATA} displays an \ATA. Observe that
the \emph{arc} starting from $\ell_0$ has two destinations: $\ell_0$
and $\ell_1$. When the automaton is in $\ell_0$ with clock valuation
$v$, and reads a $\sigma$, it spawns two copies of itself: the first
reads the suffix of the word from $(\ell_0,v)$, and the latter from
$(\ell_1,0)$ (observe that the clock is reset on the branch to
$\ell_1$). 
Then, every MITL formula $\Phi$ can be translated into an \ATA
$\Aa_\Phi$ that recognises its models \cite{OW05}. The translation has
the advantage of being very simple and elegant, and the size of
$\Aa_\Phi$ is linear in the size of $\Phi$. Unfortunately, 
one cannot bound a priori the number of clock copies that need to be
remembered at all times along runs of an \ATA. Hence, \ATA cannot, in
general, be translated to timed automata \cite{LW08}. Moreover, the
model-checking algorithm of \cite{OW05} relies on well-quasi ordering
to ensure termination, and has non-primitive recursive complexity.

In the present work, we exploit the translation of MITL formulas into
\ATA \cite{OW05} to devise new, optimal, and -- hopefully -- elegant
and simple algorithms to translate an MITL formula into a timed
automaton. To achieve this, we rely on two technical ingredients. We
first propose (in Section~\ref{sec:an-interv-semant}) a novel
\emph{interval-based} semantics for \ATA. In this semantics, clock
valuations can be regarded as \emph{intervals} instead of \emph{single
  points}, thus our semantics generalises the standard one
\cite{OW05}. Intuitively, a state $(\ell, I)$ of an \ATA in the
interval-based semantics (where $\ell$ is a location and $I$ is an
interval) can be regarded as an \emph{abstraction} of all the
(possibly unbounded) sets of states
$\{(\ell,v_1),(\ell,v_2),\ldots,(\ell,v_n)\}$ of the standard
semantics with $v_i\in I$ for all $i$. Then, we introduce a family of
so-called \emph{approximation function} that, roughly speaking,
associate with each configuration $C$ of the \ATA in the
interval-based semantics, a set of configurations that are obtained
from $C$ by merging selected intervals in $C$. We rely on
approximation functions to \emph{bound} the number of clock copies
that are present in all configurations. Our main technical
contribution (Section~\ref{sec:from-mitl-timed}) then consists in
showing that, when considering an \ATA $\Aa_\Phi$ obtained from an
MITL formula $\Phi$, combining the interval semantics and a
well-chosen approximation function is \emph{sound}, in the sense that
the resulting semantics recognises $L(\Phi)$, while requesting only
\emph{a bounded number of clock copies}. Thanks to this result, we
provide an algorithm that translates the \ATA $\Aa_\Phi$ into a plain
timed automaton that accepts the same language.

From our point of view, the \emph{benefits} of this new approach are
as follows. From the \emph{theoretical point of view}, our
construction is the first that relies on \ATA to translate MITL
formulas into timed automata. We believe our construction is easier to
describe (and thus, hopefully, easier to implement) than the previous
approaches. The translation from MITL to \ATA is very
straightforward. The intuitions behind the translation of the \ATA
into a timed automaton are also quite natural (although the proof of
correctness requires some technicalities). From the \emph{practical
  point of view}, our approach allows us, as we briefly sketch in
Section~\ref{sec:future-works:-toward}, to envision \emph{efficient
  model checking algorithms for MITL}, in the same spirit of the
\emph{antichain approach} \cite{DDMR08} developed for LTL model
checking. Note that the key ingredient to enable this \emph{antichain}
approach is the use of \emph{alternating automata} to describe the LTL
formula. Our contribution thus lay the \emph{necessary theoretical
  basis} to enable a similar approach in a real-time setting.

{\bf Remark} Owing to lack of space, most of the proof are in the
appendix.


\section{Preliminaries}
\paragraph{Basic notions.} Let $\R$ ($\R^+$, $\N$) denote resp. the
sets of real (non-negative real, natural) numbers. We call
\textbf{\textit{interval}} a convex subset of $\R$. We rely on the
classical notation $\langle a,b\rangle$ for intervals, where $\langle$
is $($ or $[$, $\rangle$ is $)$ or $]$, $a\in\R$ and
$b\in\R\cup\{+\infty\}$. For an interval $I=\langle a,b\rangle$, we
let $\inf(I)=a$ be the \emph{infimum} of $I$, $\sup(I)=b$ be its
\emph{supremum} ($a$ and $b$ are called the \emph{endpoints} of $I$)
and $\vert I \vert = \sup(I) - \inf(I)$ be its \emph{length}.  We note
$\mathcal{I}(\R)$ the set of all intervals.  Similarly, we note
$\mathcal{I}(\R^{+})$ (resp. $\mathcal{I}(\R_{\N}))$ the set of all
intervals whose endpoints are in $\R^{+}$ (resp. in $\N \cup \lbrace
+\infty \rbrace$).  Let $I \in \mathcal{I}(\R)$ and $t \in \R$, we
note $I+t$ for $\lbrace i+t \in\R \mid i \in I \rbrace$.  Let $I$ and
$J$ be two intervals, we let $I < J$ iff $\forall i \in I, \forall j
\in J : i < j$.  For $I\in \mathcal{I}(\R)$, $v \in \R$ and ${\bowtie}
\in \lbrace <, > \rbrace$, we note: $I \bowtie v$ iff $\forall i \in
I, i \bowtie v$.

Let $\Sigma$ be a finite alphabet. A \emph{word} on a set $S$ is a
finite sequence $s=s_1\ldots s_n$ of elements in $S$. We denote by
$\leng{s}=n$ the length of $s$.  A time sequence $\bar{\tau} =
\tau_{1} \tau_{2} \tau_{3} \ldots \tau_{n}$ is a word on $\R^+$
s.t. $\forall i < \leng{\bar{\tau}}, \tau_{i} \leq \tau_{i+1}$. A
\emph{timed word} over $\Sigma$ is a pair $\theta =
(\bar{\sigma},\bar{\tau})$ where $\bar{\sigma}$ is a word over
$\Sigma$, $\bar{\tau}$ a time sequence and
$\leng{\bar{\sigma}}=\leng{\bar{\tau}}$. We also note $\theta$ as
$(\sigma_{1},\tau_{1}) (\sigma_{2},\tau_{2}) (\sigma_{3},\tau_{3})
\ldots (\sigma_{n},\tau_{n})$, and let $\leng{\theta} = n$. A
\emph{timed language} is a (possibly infinite) set of timed words.

\paragraph{Metric Interval Time Logic.} Given a finite alphabet
$\Sigma$, the formulas of MITL are defined by the following grammar,
where $\sigma \in \Sigma$, $I \in \mathcal{I}(\R_{\N})$ :
\begin{center}
$\varphi$ := $\; \top \quad \vert \quad \sigma \quad \vert \quad \varphi_{1} \wedge \varphi_{2} \quad \vert \quad \neg \varphi \quad \vert \quad \varphi_{1} U_{I} \varphi_{2}$.
\end{center}
We rely on the following usual shortcuts $\lozenge_{I} \varphi$ stands
for $\top U_{I} \varphi$, $\Box_{I} \varphi$ for $\neg \lozenge_{I}
\neg \varphi$, $\varphi_{1} \tilde{U}_{I} \varphi_{2}$ for $\neg (
\neg \varphi_{1} U_{I} \neg \varphi_{2})$, $\Box \varphi$ for
$\Box_{[0,\infty)} \varphi$ and $\lozenge \varphi$ for
$\lozenge_{[0,\infty)} \varphi$.

Given an MITL formula $\Phi$, we note $Sub(\Phi)$ the set of all
subformulas of $\Phi$, i.e. : $\sub{\Phi}=\{\Phi\}$ when
$\Phi\in\{\top\}\cup\Sigma$,
$\sub{\neg\varphi}=\{\neg\varphi\}\cup\sub{\varphi}$ and
$\sub{\Phi}=\{\Phi\}\cup\sub{\varphi_1}\cup\sub{\varphi_2}$ when
$\Phi=\varphi_1U_I\varphi_2$ or $\Phi=\varphi_1\wedge\varphi_2$. We
let $|\Phi|$ denote the \emph{size of $\Phi$}, defined as the
\emph{number} of $U$ or $\tilde{U}$ modalities it contains.

\begin{definition}[Semantics of MITL] Given a timed word $\theta =
  (\bar{\sigma},\bar{\tau})$ over $\Sigma$, a position $1\leq i\leq
  \leng{\theta}$ and an MITL formula $\Phi$, we say that $\theta$
  satisfies $\Phi$ from position $i$, written $(\theta, i) \models
  \Phi$ iff the following holds :
  \begin{itemize}
  \item $(\theta,i) \models \top$
  \item $(\theta,i) \models \sigma \Leftrightarrow \sigma_{i} = \sigma$
  \item $(\theta,i) \models \varphi_{1} \wedge \varphi_{2}
    \Leftrightarrow (\theta,i) \models \varphi_{1} \text{ and }
    (\theta,i) \models \varphi_{2}$
  \item $(\theta,i) \models \neg \varphi \Leftrightarrow (\theta,i)
    \nvDash \varphi$
  \item $(\theta,i) \models \varphi _{1}U_{I} \varphi _{2}
    \Leftrightarrow \exists i \leq j \leq \vert \theta \vert, \text{
      such that } (\theta,i) \models \varphi _{2}, \tau_{j}-\tau_{i}
    \in I \text{ and } \forall i \leq k < j, (\theta,k) \models
    \varphi_{1}$
  \end{itemize}
  We say that \textbf{\textit{$\theta$ satisfies $\Phi$}}, written
  $\theta \models \Phi$, iff $(\theta,1) \models \Phi$.  We note
  $\sem{\Phi} = \{\theta \mid \theta \models \Phi\}$.
\end{definition}
Observe that, for all MITL formula $\Phi$, $\sem{\Phi}$ is a timed
language and that we can transform any MITL formula in an equivalent
MITL formula in \emph{negative normal form} (in which negation can
only be present on letters $\sigma \in \Sigma$) using the operators :
$\wedge, \vee, \neg, U_{I}$ and $\tilde{U}_{I}$.

\begin{example}
  We can express the fact that `\textit{every occurrence of $p$ is
    followed by an occurrence of $q$ between 2 and 3 time units
    later}' by: $\Box (p \Rightarrow \lozenge_{[2,3]} q)$. Its
  negation, $\neg \big(\Box (p \Rightarrow \lozenge_{[2,3]} q )
  \big)$, is equivalent to the following negative normal form formula:
  $\top U_{[0,+\infty)} ( p \wedge \perp \tilde{U}_{[2,3]} \neg q )$.
\end{example}

\paragraph{Alternating timed automata.} Let us now recall \cite{OW07}
the notion of (one clock) \emph{alternating timed automaton} (\ATA for
short). As we will see, \ATA define timed languages, and we will use
them to express the semantics of MITL formula. Let $\Gamma(L)$ be a
set of formulas defined by the following grammar:
\begin{center}
  $\gamma$ := $\top$ $\;$ $\vert$ $\;$ $\perp$ $\;$ $\vert$ $\;$
  $\gamma_{1}$ $\vee$ $\gamma_{2}$ $\;$ $\vert$ $\;$ $\gamma_{1}$
  $\wedge$ $\gamma_{2}$ $\;$ $\vert$ $\;$ $\ell$ $\;$ $\vert$ $\;$ $x
  \bowtie c$ $\;$ $\vert$ $\;$ $x.\gamma$
\end{center}
where $c \in \N$, ${\bowtie} \in \lbrace <, \leq, >, \geq \rbrace \text{
  and } \ell \in L$.  We call $x \bowtie c$ a \emph{clock
  constraint}. Intuitively, the expression $x.\gamma$ means that
clock $x$ must be reset to 0.
\begin{definition}[\cite{OW07}]
  A \textbf{\textit{one-clock alternating timed automaton}} (\ATA) is
  a tuple $\mathcal{A} = (\Sigma, L, \ell_{0}, F, \delta)$ where
  $\Sigma$ is a finite alphabet, $L$ is a finite set of locations,
  $\ell_{0}$ is the initial location, $F \subseteq L$ is a set of
  accepting locations, $\delta : L \times \Sigma \rightarrow
  \Gamma(L)$ is the transition function.
\end{definition}
We assume that, for all $\gamma_1$, $\gamma_2$ in $\Gamma(L)$:
$x.(\gamma_{1} \vee \gamma_{2}) = x.\gamma_{1} \vee x.\gamma_{2}$,
$x.(\gamma_{1} \wedge \gamma_{2}) = x.\gamma_{1} \wedge x.\gamma_{2}$,
$x.x.\gamma = x.\gamma$, $x.(x \bowtie c) = 0 \bowtie c$, $x.\top =
\top$ and $x.\perp = \perp$. Thus, we can write any formula of
$\Gamma(L)$ in disjunctive normal form, and, from now on, we assume
that $\delta (\ell,\sigma)$ is written in disjunctive normal
form. That is, for all $\ell$, $\sigma$, we have $\delta (\ell,\sigma)
= \underset{j}{\bigvee} \underset{k}{\bigwedge} A_{j,k}$, where each
term $A_{j,k}$ is of the form $\ell$, $x.\ell$, $x \bowtie c$ or $0
\bowtie c$, with $\ell \in L$ and $c \in \N$.  We call \emph{arc} of
the \ATA $\mathcal{A}$ a triple $(\ell, \sigma, \bigwedge_k A_{j,k})$
s.t. $\bigwedge_k A_{j,k}$ is a disjunct in $\delta (\ell,\sigma)$.

\begin{example}
\label{ex1}
As an example, consider the \ATA $\mathcal{A}$ in Fig.~\ref{ExATA},
over the alphabet $\Sigma = \lbrace \sigma \rbrace$.  $\mathcal{A}$
has three \emph{locations} $\ell_{0}$, $\ell_{1} \text{ and }
\ell_{2}$, such that $\ell_{0}$ is initial and $\ell_{0}$ and
$\ell_{1}$ are final. $\mathcal{A}$ has a unique clock $x$ and its
transition function is given by : $\delta (\ell_{0},\sigma) = \ell_{0}
\wedge x.\ell_{1}$, $\delta (\ell_{1},\sigma) = (\ell_{2} \wedge x=1)
\vee (\ell_{1} \wedge x\neq 1)$ and $\delta (\ell_{2},\sigma) =
\ell_{2}$.  The \emph{arcs} of $\mathcal{A}$ are thus ($\ell_{0},
\sigma, \ell_{0} \wedge x.\ell_{1})$, $(\ell_{1}, \sigma, \ell_{2}
\wedge x=1)$, $(\ell_{1}, \sigma, \ell_{1} \wedge x\neq 1)$ and
$(\ell_{2}, \sigma, \ell_{2})$.  Observe that, in the figure we
represent the (conjunctive) arc $(\ell_{0}, \sigma, \ell_{0} \wedge
x.\ell_{1})$ by an arrow splitting in two branches connected resp. to
$\ell_{0}$ and $\ell_{1}$ (possibly with different resets: the reset
of clock $x$ is depicted by $x:=0$).  Intuitively, \emph{taking the
  arc $(\ell_{0}, \sigma, \ell_{0} \wedge x.\ell_{1})$ } means that,
when reading a $\sigma$ from location $\ell_0$ and clock value $v$,
the automaton should start \emph{two copies of itself}, one in
location $\ell_0$, with clock value $v$, and a second in location
$\ell_1$ with clock value $0$. Both copies should accept the suffix
for the word to be accepted. This notion will be defined formally in
the next section.
\end{example}

\begin{figure}[t]
\centering
\begin{tikzpicture}[xscale=1]

\draw (0,0) node [circle, double, draw, inner sep=3pt] (A) {$\ell_{0}$};
\draw (3,0) node [circle, double, draw, inner sep=3pt] (B) {$\ell_{1}$};
\draw (6,0) node [circle, draw, inner sep=3pt] (C) {$\ell_{2}$};

\draw [-latex'] (-1,0) -- (A) ;
\draw [-latex'] (B) .. controls +(310:1.3cm) and +(230:1.3cm) .. (B) node [pos=0.5,below] {$\sigma, x \neq 1$};
\draw [-latex'] (C) .. controls +(310:1.3cm) and +(230:1.3cm) .. (C) node [pos=0.5,below] {$\sigma$};
\draw [-latex'] (B) -- (C) node [pos=.5,above] {$\sigma, x = 1$};
\draw[-*] (A) -- (1.5,0) node [pos=.3,above] {$\sigma$};
\draw[-latex'] (1.5,0) -- (B) node [pos=.5,above] {$x := 0$};
\draw [-latex'] (1.45,0) .. controls +(230:0.8cm) .. (A);

\end{tikzpicture}
\caption{\ATA $\mathcal{A}$}
\label{ExATA}
\end{figure}

\section{An intervals semantics for \ATA\label{sec:an-interv-semant}}
The \emph{standard} semantics for \ATA \cite{OW05,LW08} is defined as
an infinite transition system whose \emph{configurations} are finite
sets of pairs $(\ell,v)$, where $\ell$ is a location and $v$ is the
valuation of the (unique) clock. Intuitively, each configuration thus
represents the current state of all the copies (of the unique clock)
that run in parallel in the \ATA. The transition system is
\emph{infinite} because one cannot bound, a priori, the number of
different clock valuations that can appear in a single configuration,
thereby requiring peculiar techniques, such as well-quasi orderings
(see \cite{OW07}) to analyse it. In this section, we introduce a
\emph{novel} semantics for \ATA, in which \emph{configurations} are
sets of \emph{states} $(\ell, I)$, where $\ell$ is a location of the
\ATA and $I$ is an \emph{interval}, instead of a single point in
$\R^+$. Intuitively, a state $(\ell, I)$ is an abstraction of all the
states $(\ell,v)$ with $v\in I$, in the standard semantics. We further
introduce the notion of \emph{approximation function}. Roughly
speaking, an approximation function associates with each configuration
$C$ (in the interval semantics), a set of configurations that
\emph{approximates} $C$ (in a sense that will be made precise later),
\emph{and contains less states} than $C$. In
section~\ref{sec:from-mitl-timed}, we will show that the interval
semantics, combined to a proper approximation function, allows us to
build, from all MITL formula $\Phi$, an \ATA $\Aa_\Phi$ accepting
$\sem{\Phi}$, and whose \emph{reachable configurations contain a
  bounded number of intervals}. This will be the basis of our
algorithm to build a \emph{timed automaton} recognising $\Phi$ (and
hence performing automata-based model-checking of MITL).

We call \emph{state} of an \ATA $\mathcal{A} = (\Sigma, L, \ell_{0},
F, \delta)$ a couple $(\ell,I)$ where $\ell \in L$ and $I \in
\mathcal{I}(\R^{+})$.  We note $S = L \times \mathcal{I}(\R^{+})$ the
state space of $\mathcal{A}$.  A state $(\ell,I)$ is \emph{accepting}
iff $\ell \in F$.  When $I = [v,v]$ (sometimes denoted $I=\{v\}$), we
shorten $(\ell,I)$ by $(\ell,v)$. A \emph{configuration} of an \ATA
$\mathcal{A}$ is a (possibly empty) finite set of states of
$\mathcal{A}$ whose intervals associated to a same location 
are disjoint.  In the rest of the paper,
we sometimes see a configuration $C$ as a function from $L$ to
$2^{\Ii(\R^+)}$ s.t. for all $\ell\in L$: $C(\ell)=\{I\mid (\ell,I)\in
C\}$. We note $\configs{\Aa}$ the set of all configurations of
$\mathcal{A}$.  The \emph{initial configuration} of $\mathcal{A}$ is
$\lbrace(\ell_{0},0)\rbrace$. A configuration is \emph{accepting} iff
all the states it contains are accepting (in particular, the empty
configuration is accepting).  For a configuration $C$ and a delay $t$
$\in \R^{+}$, we note $C+t$ the configuration $\lbrace (\ell,I+t)
\vert (\ell,I) \in C \rbrace$. From now on, we assume that, for all
configurations $C$ and all locations $\ell$: when writing $C(\ell)$ as
$\{I_1,\ldots, I_m\}$ we have $ I_{i} < I_{i+1}$ for all $1 \leq i <
m$. Let $E$ be a finite set of intervals from
$\mathcal{I}(\R^{+})$. We let $\nbclocks{E}=|\{[a,a]\in E\}|+2\times
|\{I\in E\mid \inf(I)\neq\sup(I)\}$ denote the \emph{number of clock
  copies} of $E$. Intuitively, $\nbclocks{E}$ is the number of
\emph{individual clocks} we need to encode all the information present
in $E$, using one clock to track singular intervals, and two clocks to
retain $\inf(I)$ and $\sup(I)$ respectively for non-singular intervals
$I$. For a configuration $C$, we let $\nbclocks{C}=\sum_{\ell\in L}
\nbclocks{C(\ell)}$.

\paragraph{Interval semantics.} Our definition of the \emph{interval
  semantics} for \ATA follows the definition of the \emph{standard}
semantics as given by Ouaknine and Worrell \cite{OW05}, adapted to
cope with intervals. Let $M\in\configs{\Aa}$ be a configuration of an
\ATA \Aa, and $I \in \mathcal{I}(\R^{+})$. We define the satisfaction
relation "$\models_{I}$" on $\Gamma(L)$~as:
$$
\begin{array}{ll}
  M \models_{I} \top\\
  M \models_{I} \gamma_{1} \wedge \gamma_{2} &\text{ iff } M \models_{I} \gamma_{1} \text{ and } M \models_{I} \gamma_{2}\\
  M \models_{I} \gamma_{1} \vee \gamma_{2} &\text{ iff } M \models_{I} \gamma_{1} \text{ or } M \models_{I} \gamma_{2}
\end{array}
\qquad
\begin{array}{ll}
  M \models_{I} \ell  &\text{ iff } (\ell,I) \in M\\
  M \models_{I} x \bowtie c &\text{ iff } \forall x \in I, x \bowtie c\\
  M \models_{I} x.\gamma &\text{ iff } M \models_{[0,0]} \gamma
\end{array}
$$
We say that $M$ is a \textit{minimal model} of the formula $\gamma$
$\in$ $\Gamma(L)$ with respect to the interval $I \in
\mathcal{I}(\R^{+}) \text{ iff } M \models_{I} \gamma$ and there is no
$M' \subsetneq M$ such that ${M' \models_{I} \gamma}$.  Remark that a
formula $\gamma$ can admit several minimal models (one for each
disjunct in the case of a formula of the form
$\gamma=\underset{j}{\bigvee} \underset{k}{\bigwedge} A_{j,k}$).
Intuitively, for $\ell \in L, \sigma \in \Sigma$ and $I \in
\mathcal{I}(\R^{+})$, a minimal model of $\delta(\ell,\sigma)$ with
respect to $I$ represents a configuration the automaton can reach from
state $(\ell,I)$ by reading $\sigma$. The definition of $M \models_{I}
x \bowtie c$ only allows to take a transition $\delta(\ell,\sigma)$
from state $(\ell,I)$ if all the values in $I$ satisfy the clock
constraint $x \bowtie c$ of $\delta(\ell,\sigma)$.

\begin{example}
  Let us consider again the \ATA of Fig.~\ref{ExATA}. A minimal
  model $M$ of $\delta(\ell_{1}, \sigma)$ with respect to [1.5,2] must
  be such that : $M \models_{[1.5,2]} (\ell_{1} \wedge x \neq 1) \vee
  (\ell_{2} \wedge x = 1)$.  As $\exists v \in [1.5,2]$ s.t. $v \neq
  1$, it is impossible that $M \models_{[1.5,2]} x = 1$. However, as
  $\forall v \in [1.5,2], v \neq 1$, $M \models_{[1.5,2]} x \neq 1$
  and so $M \models_{[1.5,2]} (\ell_{1} \wedge x \neq 1) \vee
  (\ell_{2} \wedge x = 1)$ iff $M \models_{[1.5,2]} \ell_{1}$,
  i.e. $(\ell_{1}, [1.5,2]) \in M$.  So, $\lbrace (\ell_{1},
  [1.5,2])\rbrace$ is the unique \emph{minimal model} of
  $\delta(\ell_{1}, \sigma)$ wrt $[1.5,2]$.
\end{example}

\paragraph{Approximation functions.} As stated before, our goal is to
define a semantics for \ATA that enables to bound the number of clock
copies. To this end, we define the notion of approximation function:
we will use such functions to reduce the number of clock copies
associated with each location in a configuration.  An approximation
function associates with each configuration $C$ a set of
configurations $C'$ s.t. $\nbclocks{C'(\ell)}\leq\nbclocks{C(\ell)}$
and s.t. the intervals in $C'(\ell)$, \emph{cover} those of $C(\ell)$,
for all $\ell$. Then, we define the semantics of an \ATA $\mathcal{A}$
by means of a transition system $\mathcal{T}_{\mathcal{A}, f}$ whose
definition is \emph{parametrised by an approximation function $f$}.

\begin{definition}
  Let $\Aa$ be an \ATA $\Aa$. An \emph{approximation function} is a
  function $f:\configs{\Aa}\mapsto 2^{\configs{\Aa}}$ s.t. for all
  configurations $C$, for all $C'\in f(C)$, for all locations $\ell\in
  L$: $(i)$ $\nbclocks{C'(\ell)}\leq\nbclocks{C(\ell)}$, $(ii)$ for
  all $I\in C(\ell)$, there exists $J\in C'(\ell)$ s.t. $I \subseteq
  J$, $(iii)$ for all $J \in C'(\ell)$, there are $I_1, I_2 \in
  C(\ell)$ s.t. $\inf(J) = \inf(I_1)$ and $\sup(J) = \sup(I_2)$.  We
  note $\appf{\Aa}$ the set of approximation functions for~$\Aa$.
\end{definition}

\begin{definition}
  Let $\mathcal{A}$ be an \ATA and let $f\in\appf{\Aa}$ be an
  approximation function.  The \emph{$f$-semantics of $\Aa$} is the
  transition system $\mathcal{T}_{\mathcal{A}, f} = (\configs{\Aa},
  \rightsquigarrow,\break \longrightarrow_{f})$ on configurations of
  $\mathcal{A}$ defined as follows:
  \begin{itemize}
  \item the transition relation $\rightsquigarrow$ takes care of the
    elapsing of time : $\forall t\in \R^{+}, C
    \overset{t}{\rightsquigarrow} C' \text{ iff } C' = C+t$. We let
    ${\rightsquigarrow} = \underset{t \in \R^{+}}{\bigcup}
    \overset{t}{\rightsquigarrow}$.
  \item the transition relation $\longrightarrow$ takes care of
    discrete transitions between locations and of the approximation :
    $C = \lbrace(\ell_{k},I_{k})_{k \in K} \rbrace$
    $\overset{\sigma}{\longrightarrow} C'$ iff there exists a
    configuration $C"=\underset{k \in K}{\bigcup} M_{k}$ s.t. $(i)$
    for all $k$: $M_{k}$ is a minimal model of
    $\delta(\ell_{k},\sigma)$ with respect to $I_{k}$, and $(ii)$ $C'
    \in f(C")$.  We let $\longrightarrow_{f} = \underset{\sigma \in
      \Sigma}{\bigcup} \overset{\sigma}{\longrightarrow_{f}}$.
  \end{itemize}
\end{definition}

We can now define the accepted language of an \ATA (parametrised by an
approximation function $f$). Let $\theta = (\bar{\sigma},\bar{\tau})$
be a timed word s.t.  $\leng{\theta}=n$, and let $f\in\appf{\Aa}$ be
an approximation function.  Let us note $t_{i} = \tau_{i} -
\tau_{i-1}$ for all $1 \leq i \leq \vert \theta \vert$, assuming
$\tau_{0} = 0$.  An \emph{$f$-run} of $\Aa$ on $\theta$ is a finite
sequence of discrete and continuous transitions in $\Tt_{\Aa,f}$ that
is labelled by $\theta$, i.e.  a sequence of the form: $ C_{0}
\overset{t_{1}}{\rightsquigarrow} C_{1}
\overset{\sigma_{1}}{\longrightarrow_{f}} C_{2}
\overset{t_{2}}{\rightsquigarrow} C_{3}
\overset{\sigma_{2}}{\longrightarrow_{f}}
... \overset{t_{n}}{\rightsquigarrow} C_{2n-1}
\overset{\sigma_{n}}{\longrightarrow_{f}} C_{2n} $.
We say that an $f$-run is \emph{accepting} iff its last configuration
$C_{2n}$ is accepting and we say that a timed word is $f$-accepted by
$\mathcal{A}$ iff there exists an accepting $f$-run of $\mathcal{A}$
on this word.  We note $L_{f}(\mathcal{A})$ the language of all finite
timed words $f$-accepted by $\mathcal{A}$. In the reset of the paper,
we (sometimes) use the abbreviation $C_{i} \overset{t,
  \sigma}{\longrightarrow_{f}} C_{i+2}$ for $C_{i}
\overset{t}{\rightsquigarrow} C_{i+1} = C_{i}+t
\overset{\sigma}{\longrightarrow_{f}} C_{i+2}$.

Observe that this interval semantics generalises the standard \ATA
semantics \cite{OW05}. This standard semantics can be recovered by
considering $\Tt_{\Aa, Id}$, where $Id$ is the approximation function
such that $Id(C) = \{C\}$ for all $C$. Indeed, in $\Tt_{\Aa, Id}$, all
the reachable configurations contain only states of the form $(\ell,
[a,a])$, i.e., all intervals are singular. So, each state $(\ell,
[a,a])$ can be naturally mapped to a state $(\ell, a)$ in the standard
semantics. From now on, we denote $L_{Id}(\Aa)$ by~$L(\Aa)$.
\begin{example}
  Let us consider again the \ATA $\Aa$ in Fig.~\ref{ExATA}, and the
  timed word $\theta = (\sigma,0) (\sigma,0.2) (\sigma,0.5)$, with
  $\leng{\theta}=3$. Let $f$ be the approximation function s.t. for
  all $C\in\configs{\Aa}$: $f(C)=\big\{C(\ell_0)\cup C(\ell_2)\cup
  \{(\ell_1,[inf(I_1),sup(I_m)])\}\big\}$ if
  $C(\ell_1)=\{I_1,I_2,\ldots I_m\}\neq \emptyset$ (assuming, as
  mentioned before, that $I_1 < I_2 < \cdots < I_m$); and $f(C)=\{C\}$
  if $C(\ell_1)=\emptyset$. Thus, roughly speaking, $f(C)$ always
  contains one configuration, which is obtained from $C$ by merging
  all the intervals in $C(\ell_1)$ and keeping the rest of the
  configuration untouched. Then, an $f$-run on $\theta$ is:
  $\rho_1= \{(\ell_{0},0)\}\xrightarrow{0,\sigma}%
  \{(\ell_0,0), (\ell_1,0)\}\xrightarrow{0.2,\sigma}%
  \{(\ell_0,0.2), (\ell_1,[0,0.2])\}\xrightarrow{0.3,\sigma}%
  \{(\ell_0,0.5), (\ell_1,[0,0.5])\}%
  $. %
  Also, an $Id$-run on $\theta$ is:
  $ \rho_2= \{(\ell_{0},0)\} %
  \xrightarrow{0,\sigma} %
  \{(\ell_0,0), (\ell_1,0)\} %
  \xrightarrow{0.2,\sigma} %
  \{(\ell_0,0.2), (\ell_1,0),\break (\ell_1,0.2)\} %
  \xrightarrow{0.3,\sigma} %
  \{(\ell_0,0.5), (\ell_1,0), (\ell_1,0.3),
  (\ell_1,0.5)\}%
  $. %
  Now, consider the timed word $\theta'=\theta (\sigma,1.1)$. An $Id$-run
  on $\theta'$ is $\rho_3=\rho_2\xrightarrow{0.6,\sigma}
  \{(\ell_0,1.1), (\ell_1,0), (\ell_1,0.6),
  (\ell_1,0.9),\break (\ell_1,1.1)\}$ (hence $\theta'$ is
  $Id$-accepted by $\Aa$), but $\Aa$ has no $f$-run on
  $\theta'$. Indeed, letting $0.6$ t.u. elapse from $\rho_1$'s last
  configuration yields $\{(\ell_0,1.1), (\ell_1,[0.6,1.1])\}$
  from which no transition can be fired, because $[0.6,1.1]$ satisfies
  neither $x\neq 1$ nor $x=1$, which are the respective guards of the
  arcs from $\ell_1$.
\end{example}

In the rest of the paper we will rely mainly on approximation
functions that enable to bound the number of clock copies in all
configurations along all runs of an \ATA $\Aa$.  Let $k \in \N$ be a
constant. We say that $f_k\in\appf{\Aa}$ is a \emph{$k$-bounded
  approximation function} iff for all $ C \in \configs{\Aa}$, for all
$C' \in f_{k}(C)$: $\nbclocks {C'} \leq k$.

\paragraph{Accepted language and approximations.} Let us now study the
relationship between the standard semantics of \ATA and the family of
semantics obtained when relying on an approximation function that is
different from $Id$. We show that introducing approximations does not
increase the accepted language:
\begin{proposition}
\label{inclu}
For all \ATA $\Aa$, for all $f \in \appf{\Aa}$: $L_{f}(\Aa) \subseteq
L(\Aa)$.
\end{proposition}
\begin{proof}[sketch]
  Let $C_0\timestep{t_1}C_1\xrightarrow{\sigma_1} C_2
  \timestep{t_2}C_3\cdots \xrightarrow{\sigma_n}C_{2n}$ be an
  \emph{accepting} $f$-run of $\Aa$ on $\theta$, and let us build,
  inductively, an accepting $Id$-run
  $D_0\timestep{t_1}D_1\xrightarrow{\sigma_1} D_2
  \timestep{t_2}D_3\cdots \xrightarrow{\sigma_n}D_{2n}$ on $\theta$
  s.t. the following \emph{invariant} holds: for all $1\leq i\leq 2n$,
  for all $(\ell, [v,v])\in D_i$, there is $(\ell, I)\in C_i$
  s.t. $v\in I$. The \emph{base case} is trivial since $C_0=D_0$. For
  the \emph{inductive case}, we first observe that the elapsing of
  time maintains the invariant. Thus, we have to show that each
  discrete step in the $f$-run can be simulated by a discrete step in
  the $Id$-run that maintains the invariant. A $\sigma$ labeled
  discrete step from some configuration $C_{2j+1}$ in the $f$-run
  consists in selecting an arc $a_s$ of the form
  $(\ell,\sigma,\gamma)$ for each $s=(\ell,I)$ in $C_{2j+1}$, whose
  guard is satisfied by $I$. Then, firing all these arcs yields a
  configuration $E$, and $C_{2j+2}\in f(E)$. From each $s'=(\ell,
  [v,v])$ in $D_{2j+1}$, we fire the arc $a_s$ where $s=(\ell,I)$ is a
  state in $C_{2j+1}$ s.t $v\in I$. Such an $s$ exists by induction
  hypothesis. Since the effects of the arcs are the same, and by
  properties of the approximation function, we conclude that
  $D_{2j+2}$ and $C_{2j+2}$ respect the invariant. In particular
  $D_{2n}$ and $C_{2n}$ respect it, hence $D_{2n}$ is accepting. 
\end{proof}

\section{From MITL to Timed Automata\label{sec:from-mitl-timed}}
In this section, we present our new technique to build, from any MITL
formula $\Phi$, a \emph{timed automaton} that accepts
$\sem{\Phi}$. Our technique relies on two ingredients. First, we
recall \cite{OW07} how to build, from all MITL formula $\Phi$, and
\ATA $\Aa_\Phi$ s.t. $L(\Aa_\Phi)=\sem{\Phi}$. This is not sufficient
to obtain a timed automaton, as, in general, the semantics of an \ATA
needs an \emph{unbounded} number of clock copies, which prevents us
from translating all \ATA into timed automata. The second ingredient
is the definition of a family of \emph{bounded approximation
  functions} $f^\star_\Phi$, s.t., for all MITL formula $\Phi$,
$L_{\appfunc{\Phi}}(\Aa_\Phi)=L(\Aa_\Phi)$. Since each $f^\star_\Phi$
is a \emph{bounded approximation function}, the number of clock copies
in the $f^\star_\Phi$-semantics of $\Aa_\Phi$ is \emph{bounded}, which
allows us to build a \emph{timed automaton $\Bb_\Phi$ with the same
  semantics} (thus, $\Bb_\Phi$ accepts $\sem{\Phi}$).

\paragraph{From MITL to \ATA.} We begin by recalling\footnote{Remark
  that in \cite{OW07}, the authors are concerned with MTL, but since
  MITL is a syntactic fragment of MTL, the procedure applies here.}
\cite{OW07} how to build, from any MITL formula $\Phi$ (in negative
normal form), an \ATA $\Aa_\Phi$ s.t. $L(\Aa_\Phi)=\sem{\Phi}$. We let
$\mathcal{A}_{\Phi} = (\Sigma, L, \ell_{0}, F, \delta)$ where: $L$ is
the set containing the initial copy of $\Phi$, noted `$\Phi_{init}$',
and all the formulas of $Sub(\Phi)$ whose outermost connective is
`$U$' or `$\tilde{U}$'; $\ell_{0} = \Phi_{init}$; $F$ is the set of
the elements of $L$ of the form $\varphi_{1} \tilde{U}_{I}
\varphi_{2}$. Finally $\delta $ is defined\footnote{Remark that the
  $x\leq sup(I)$ and $x>sup(I)$ conditions in the resp. definitions of
  $\delta(\varphi_{1}U_{I} \varphi_{2}, \sigma)$ and
  $\delta(\varphi_{1}\tilde{U}_{I} \varphi_{2}, \sigma)$ have been
  added here for technical reasons. This does not modify the accepted
  language. Indeed, in \cite{OW07}, these conditions are given in the
  infinite word semantics of \ATA.} by induction on the structure of
$\Phi$:
\begin{itemize}
\item $\delta(\Phi_{init}, \sigma) = x.\delta(\Phi, \sigma)$
\item $\delta(\varphi_{1} \vee \varphi_{2}, \sigma) =
  \delta(\varphi_{1}, \sigma) \vee \delta(\varphi_{2}, \sigma)$;  $\delta(\varphi_{1} \wedge \varphi_{2}, \sigma) =
  \delta(\varphi_{1}, \sigma) \wedge \delta(\varphi_{2}, \sigma)$
\item $\delta(\varphi_{1}U_{I} \varphi_{2}, \sigma) =
  (x.\delta(\varphi_{2},\sigma) \wedge x \in I) \vee
  (x.\delta(\varphi_{1},\sigma) \wedge \varphi_{1}U_{I} \varphi_{2}
  \wedge x\leq sup(I))$
\item $\delta(\varphi_{1} \tilde{U}_{I} \varphi_{2}, \sigma) =
  (x.\delta(\varphi_{2},\sigma) \vee x\notin I) \wedge
  (x.\delta(\varphi_{1},\sigma) \vee \varphi_{1} \tilde{U}_{I}
  \varphi_{2} \vee x > sup(I))$
\item $\forall \sigma_1,\sigma_2\in \Sigma$: $\delta(\sigma_{1},
  \sigma_{2}) = \left\{
    \begin{array}{ll}
      \text{true} &\text{if } \sigma_{1} \text{=} \sigma_{2} \\
      \text{false}&\text{if } \sigma_{1} \neq \sigma_{2}
    \end{array}
  \right.$
\item $\forall \sigma_1,\sigma_2\in \Sigma$: $\delta(\neg \sigma_{1}, \sigma_{2}) = \left\{
    \begin{array}{ll}
      \text{false} &\text{if } \sigma_{1} \text{=} \sigma_{2} \\
      \text{true} &\text{if } \sigma_{1} \neq \sigma_{2}
    \end{array}
  \right.$
\item $\forall \sigma\in\Sigma$: $\delta(\top,\sigma)=\top$ and
  $\delta(\bot,\sigma)=\bot$.
\end{itemize}
To simplify the following proofs, we deviate slightly from that
definition, and assume that if a formula of type $\varphi_{1} U_{I}
\varphi_{2}$ or $\varphi_{1} \tilde{U}_{I} \varphi_{2}$ appears more
than once as a sub-formula of $\Phi$, the occurrences of this formula
are supposed different and are encoded as different locations. With
this definition, we have:
\begin{theorem}[\cite{OW07}]
  For all MITL formula $\Phi$: $L(\Aa_\Phi)=\sem{\Phi}$.
\end{theorem}

\begin{figure}[t]
\centering
\begin{tikzpicture}[scale=.9]

\begin{scope}
\draw (0,0) node [circle, double, draw, inner sep=3pt] (A) {$\ell_{\Box}$};
\draw (3,0) node [circle, draw, inner sep=3pt] (B) {$\ell_\lozenge$};

\draw [-latex'] (-1,0) -- (A) ;
\draw [-latex'] (A) .. controls +(230:1.2cm) and +(310:1.2cm) .. (A) node [pos=0.5,below] {$b$};
\draw [-latex'] (B) .. controls +(230:1.2cm) and +(310:1.2cm) .. (B) node [pos=0.5,below] {$a,b$};
\draw [-latex'] (B) -- (5.2,0) node [pos=.5,above] {$\begin{array}{c}b\\ x \in [1,2]\end{array}$};
\draw[-*] (A) -- (1.5,0) node [pos=.3,below] {$a$};
\draw[-latex'] (1.5,0) -- (B) node [pos=.5,below] {$x := 0$};
\draw [-latex'] (1.4,0) .. controls +(-220:0.8cm) .. (A);
\end{scope}

\begin{scope}[shift={(5.7,0)},scale=.9]
\fill (0,0) circle (0.05cm);
\draw (0,-0.3) node {\scriptsize{0}};
\fill (2,0) circle (0.05cm);
\draw (2,-0.3) node {\scriptsize{1}};
\draw (4,-0.3) node {\scriptsize{2}};
\draw (6,-0.3) node {\scriptsize{3}};

\draw[-latex'] (0,0) -- (7,0);
\draw (7,-0.3) node {\tiny{time}};

\fill (0.6,0) circle (0.03cm);
\fill (0.8,0) circle (0.03cm);
\fill (1,0) circle (0.03cm);
\fill (1.2,0) circle (0.03cm);
\fill (1.4,0) circle (0.03cm);
\fill (1.6,0) circle (0.03cm);
\fill (1.8,0) circle (0.03cm);

\fill (2.2,0) circle (0.03cm);
\fill (2.4,0) circle (0.03cm);
\fill (2.6,0) circle (0.03cm);
\fill (2.8,0) circle (0.03cm);
\fill (3,0) circle (0.03cm);
\fill (3.2,0) circle (0.03cm);
\fill (3.4,0) circle (0.03cm);
\fill (3.6,0) circle (0.03cm);

\fill (4.2,0) circle (0.03cm);
\fill (4.4,0) circle (0.03cm);
\fill (4.6,0) circle (0.03cm);
\fill (4.8,0) circle (0.03cm);
\fill (5,0) circle (0.03cm);
\fill (5.2,0) circle (0.03cm);
\fill (5.4,0) circle (0.03cm);
\fill (5.6,0) circle (0.03cm);
\fill (5.8,0) circle (0.03cm);

\fill (6.2,0) circle (0.03cm);
\fill (6.4,0) circle (0.03cm);
\fill (6.6,0) circle (0.03cm);
\fill (6.8,0) circle (0.03cm);

\fill (0.2,0) circle (0.05cm);
\draw (0.2,0.3) node {$a$};
\fill (0.4,0) circle (0.05cm);
\draw (0.4,0.3) node {$a$};
\fill (3.8,0) circle (0.05cm);
\draw (3.8,0.3) node {$a$};
\fill (4,0) circle (0.05cm);
\draw (4,0.3) node {$b$};
\fill (6,0) circle (0.05cm);
\draw (6,0.3) node {$b$};

\draw[-] (0.05,0.5) -- (0.53,0.5);
\draw[-] (0.05,0.3) -- (0.05,0.5);
\draw[-] (0.53,0.3) -- (0.53,0.5);
\draw[-latex'] (0.3,0.5) .. controls +(80:0.9cm) and +(130:0.9cm) .. (4,0.5);

\draw[-] (3.65,0.5) -- (3.91,0.5);
\draw[-] (3.65,0.3) -- (3.65,0.5);
\draw[-] (3.91,0.3) -- (3.91,0.5);
\draw[-latex'] (3.8,0.5) .. controls +(80:0.9cm) and +(130:0.9cm) .. (6,0.5);
\end{scope}
\end{tikzpicture}
\caption{(left) \ATA $\mathcal{A}_{\Phi_1}$ with $\Phi_1=\Box ( a
  \Rightarrow \lozenge_{[1,2]} b)$. \ (right) The grouping of clocks. \label{ExGilles}\label{Graph1}}
\end{figure}
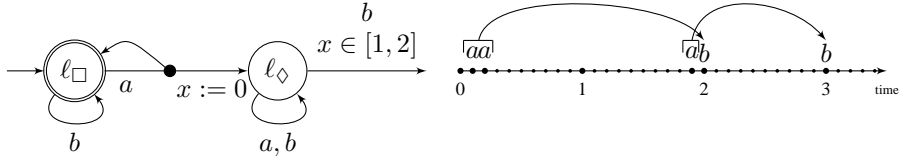

\begin{example}
  As an example consider the formula $\Phi_1=\Box( a \Rightarrow
  \lozenge_{[1,2]} b)$, which is a shorthand for
  $\bot\tilde{U}_{[0,+\infty)}\big(a\Rightarrow (\top
  U_{[1,2]}b)\big)$. The \ATA $\Aa_{\Phi_1}$ is given in
  Fig.~\ref{ExGilles} (left), where the location $\ell_\Box$ corresponds to
  $\Phi_1$ and the location $\ell_\lozenge$ corresponds to $\top
  U_{[1,2]}b$. One can check that this automaton follows strictly the
  above definition, after simplification of the formulas,
  \emph{except} that we have remove the $\Phi_{init}$ location and
  used the $\ell_\Box$ location as initial location instead, to
  enhance readability of the example (this does not modify the
  accepted language, in the present case). Observe the edge labeled by
  $b,x\in[1,2]$ from $\ell_\lozenge$, without target state: it depicts
  the fact that, after simplification: $\delta(\top U_{[1,2]}b,
  b)=(x\in[1,2])\vee (\top U_{[1,2]}b)$. Intuitively, this means that,
  when the automaton has a copy in location `$\lozenge$' with a
  clock valuation in $[1,2]$, the copy can be \emph{removed} from the
  automaton, because a minimal model of $x\in[1,2]$ wrt to a valuation
  $v$ with $v\in[1,2]$ is $\emptyset$.

  To help us build an intuition of the $\appfunc{\Phi_1}$ function,
  let us consider the $Id$-run $\rho_1$ of $\Aa_{\Phi_1}$ on
  $\theta_1=(a,0.1)(a,0.2)(a,0.3)(b,2)$ depicted in
  Fig.~\ref{reprun}.  Observe that $\theta_1\models\Phi_1$, and that
  $\rho_1$'s last configuration is indeed accepting. Also note that,
  as in the example of Fig.~\ref{ExATA}, the number of clock copies
  necessary in the $Id$-semantics cannot be bounded.  Now, let us
  discuss the intuition behind $\appfunc{\Phi_1}$ by considering
  $\theta_1$ again. Consider $\rho_1'$ the run prefix of $\rho_1$
  ending in $\{(\ell_\Box ,0.2),(\ell_\lozenge ,0),(\ell_\lozenge
  ,0.1)\}$.  Clearly, the last configuration of $\rho_1'$ can be
  \emph{over-approximated} by \emph{grouping} the two clock values $0$
  and $0.1$ into the \emph{smallest interval that contains them both},
  i.e. $[0,0.1]$. This intuitions is compatible with the definition of
  \emph{bounded approximation function}, and yields the
  \emph{accepting} run $\rho_1''$ depicted in Fig.~\ref{reprun}.
  Nevertheless, we must be be careful when grouping clock copies. Let
  us consider $\theta_2 = (a,0.1) (a,0.2) (a,1.9) (b,2) (b,3)\in
  \sem{\Phi_1}$, as witnessed by $\rho_2$ depicted in
  Fig.~\ref{reprun}. When grouping \emph{in the same interval}, the
  three clock copies created in $\ell_\lozenge$ (along $\rho_2$) by
  the reading of the three $a$'s (and letting further $0.1$ time unit
  elapse) yields the run prefix of $\rho_2'$ depicted in
  Fig.~\ref{reprun} ending in $\{(\ell_\Box ,2),(\ell_\lozenge
  ,[0.1,1.9])\} $. From the last configuration of this run, the edge
  with guard $x\in[1,2]$ and origin $\ell_\lozenge$ cannot be
  taken. Thus, the only way to extend this prefix is through $\rho_2'$
  (depicted in Fig.~\ref{reprun}) which yields a run that \emph{does
    not accept $\theta_2$}. Obviously, by grouping the two clock
  copies created in $\ell_\lozenge$ by the two first $a$'s, and by
  keeping the third one apart, one obtains the accepting run
  $\rho_2''$ (depicted in Fig.~\ref{reprun}).  Fig.~\ref{Graph1} (right)
  shows the intuition behind the grouping of clocks. The two first
  positions (with $\sigma_1=\sigma_2=a$) of the word satisfy $\Phi_1$,
  \emph{because of the $b$ in position $4$} (with $\tau_4=2$), while
  position 3 (with $\sigma_3=a$) satisfies $\Phi_1$ \emph{because of
    the $b$ in position $5$} (with $\tau_5=3$). This explains why we
  group the two first copies (corresponding to the two first $a$'s)
  and keep the third one apart.

\begin{figure}[t]
\centering
\begin{tikzpicture}[xscale=.82,yscale=.9]
\everymath{\scriptstyle}

\draw (-1.75,-.35) node [below] (Ainit) {$\rho_1$};

\draw (-1,0) node [rectangle,inner sep=1pt,rounded corners=2mm,below]
(Ainit) {$\begin{array}{l} \ell_\Box\\ \ell_\lozenge \end{array}$};

\draw (0,0) node [rectangle,fill=black!10!white,inner sep=1pt,rounded
  corners=2mm,below] (A) {$\begin{array}{l}
    \{0\}\\ \phantom{a} \end{array}$};

\draw (2,0) node [rectangle,fill=black!10!white,inner sep=1pt,rounded
  corners=2mm,below] (B) {$\begin{array}{l}
    \{0.1\}\\ \{0\} \end{array}$};

\draw (4.5,0) node [rectangle,fill=black!10!white,inner sep=1pt,rounded
  corners=2mm,below] (C) {$\begin{array}{l}
    \{0.2\}\\ \{0\},\{0.1\} \end{array}$};

\draw (7.5,0) node [rectangle,fill=black!10!white,inner
  sep=1pt,rounded corners=2mm,below] (D) {$\begin{array}{l}
    \{0.3\}\\ \{0\},\{0.1\},\{0.2\} \end{array}$};

\draw (10,0) node [rectangle,fill=black!10!white,inner sep=1pt,rounded
  corners=2mm,below] (E) {$\begin{array}{l}
    \{2\}\\ \phantom{a}  \end{array}$};

\draw [-latex'] (A) -- (B) node [pos=0.5,above] {$0.1,a$};
\draw [-latex'] (B) -- (C) node [pos=0.5,above] {$0.1,a$};
\draw [-latex'] (C) -- (D) node [pos=0.5,above] {$0.1,a$};
\draw [-latex'] (D) -- (E) node [pos=0.5,above] {$1.7,b$};

\begin{scope}[yshift=-1.5cm]

\draw (-1.75,-.25) node [below] (Ainit) {$\rho_1''$};

\draw (-1,0) node [rectangle,inner sep=1pt,rounded corners=2mm,below]
(Ainit) {$\begin{array}{l} \ell_\Box\\ \ell_\lozenge \end{array}$};

\draw (0,0) node [rectangle,fill=black!10!white,inner sep=1pt,rounded
  corners=2mm,below] (A) {$\begin{array}{l}
    \{0\}\\ \phantom{a} \end{array}$};

\draw (2,0) node [rectangle,fill=black!10!white,inner sep=1pt,rounded
  corners=2mm,below] (B) {$\begin{array}{l}
    \{0.1\}\\ \{0\} \end{array}$};

\draw (4.5,0) node [rectangle,fill=black!10!white,inner sep=1pt,rounded
  corners=2mm,below] (C) {$\begin{array}{l}
    \{0.2\}\\ \left[0,0.1\right] \end{array}$};

\draw (7.5,0) node [rectangle,fill=black!10!white,inner
  sep=1pt,rounded corners=2mm,below] (D) {$\begin{array}{l}
    \{0.3\}\\ \left[0,0.2\right] \end{array}$};

\draw (10,0) node [rectangle,fill=black!10!white,inner sep=1pt,rounded
  corners=2mm,below] (E) {$\begin{array}{l}
    \{2\}\\ \phantom{a} \end{array}$};

\draw [-latex'] (A) -- (B) node [pos=0.5,above] {$0.1,a$};
\draw [-latex'] (B) -- (C) node [pos=0.5,above] {$0.1,a$};
\draw [-latex'] (C) -- (D) node [pos=0.5,above] {$0.1,a$};
\draw [-latex'] (D) -- (E) node [pos=0.5,above] {$1.7,b$};
\end{scope}

\begin{scope}[yshift=-3cm]

\draw (-1.75,-.25) node [below] (Ainit) {$\rho_2$};

\draw (-1,0) node [rectangle,inner sep=1pt,rounded corners=2mm,below]
(Ainit) {$\begin{array}{l} \ell_\Box\\ \ell_\lozenge \end{array}$};

\draw (0,0) node [rectangle,fill=black!10!white,inner sep=1pt,rounded
  corners=2mm,below] (A) {$\begin{array}{l}
    \{0\}\\ \phantom{a} \end{array}$};

\draw (2,0) node [rectangle,fill=black!10!white,inner sep=1pt,rounded
  corners=2mm,below] (B) {$\begin{array}{l}
    \{0.1\}\\ \{0\} \end{array}$};

\draw (4.5,0) node [rectangle,fill=black!10!white,inner sep=1pt,rounded
  corners=2mm,below] (C) {$\begin{array}{l}
    \{0.2\}\\ \{0\},\{0.1\} \end{array}$};

\draw (7.5,0) node [rectangle,fill=black!10!white,inner
  sep=1pt,rounded corners=2mm,below] (D) {$\begin{array}{l}
    \{1.9\}\\ \{0\},\{1.7\},\{1.8\} \end{array}$};

\draw (10,0) node [rectangle,fill=black!10!white,inner sep=1pt,rounded
  corners=2mm,below] (E) {$\begin{array}{l}
    \{2\}\\ \{0.1\} \end{array}$};

\draw (12,0) node [rectangle,fill=black!10!white,inner sep=1pt,rounded
  corners=2mm,below] (F) {$\begin{array}{l}
    \{3\}\\ \phantom{a} \end{array}$};

\draw [-latex'] (A) -- (B) node [pos=0.5,above] {$0.1,a$};
\draw [-latex'] (B) -- (C) node [pos=0.5,above] {$0.1,a$};
\draw [-latex'] (C) -- (D) node [pos=0.5,above] {$1.7,a$};
\draw [-latex'] (D) -- (E) node [pos=0.5,above] {$0.1,b$};
\draw [-latex'] (E) -- (F) node [pos=0.5,above] {$1,b$};
\end{scope}

\begin{scope}[yshift=-4.5cm]

\draw (-1.75,-.25) node [below] (Ainit) {$\rho'_2$};

\draw (-1,0) node [rectangle,inner sep=1pt,rounded corners=2mm,below]
(Ainit) {$\begin{array}{l} \ell_\Box\\ \ell_\lozenge \end{array}$};

\draw (0,0) node [rectangle,fill=black!10!white,inner sep=1pt,rounded
  corners=2mm,below] (A) {$\begin{array}{l}
    \{0\}\\ \phantom{a} \end{array}$};

\draw (2,0) node [rectangle,fill=black!10!white,inner sep=1pt,rounded
  corners=2mm,below] (B) {$\begin{array}{l}
    \{0.1\}\\ \{0\} \end{array}$};

\draw (4.5,0) node [rectangle,fill=black!10!white,inner sep=1pt,rounded
  corners=2mm,below] (C) {$\begin{array}{l}
    \{0.2\}\\ \left[0,0.1\right] \end{array}$};

\draw (7.5,0) node [rectangle,fill=black!10!white,inner
  sep=1pt,rounded corners=2mm,below] (D) {$\begin{array}{l}
    \{1.9\}\\ \left[0,1.8\right] \end{array}$};

\draw (10,0) node [rectangle,fill=black!10!white,inner sep=1pt,rounded
  corners=2mm,below] (E) {$\begin{array}{l}
    \{2\}\\ \left[0.1,1.9\right] \end{array}$};

\draw (12,0) node [rectangle,fill=black!10!white,inner sep=1pt,rounded
  corners=2mm,below] (F) {$\begin{array}{l}
    \{3\}\\ \left[2.1,2.9\right] \end{array}$};

\draw [-latex'] (A) -- (B) node [pos=0.5,above] {$0.1,a$};
\draw [-latex'] (B) -- (C) node [pos=0.5,above] {$0.1,a$};
\draw [-latex'] (C) -- (D) node [pos=0.5,above] {$1.7,a$};
\draw [-latex'] (D) -- (E) node [pos=0.5,above] {$0.1,b$};
\draw [-latex'] (E) -- (F) node [pos=0.5,above] {$1,b$};
\end{scope}

\begin{scope}[yshift=-6cm]

\draw (-1.75,-.25) node [below] (Ainit) {$\rho''_2$};

\draw (-1,0) node [rectangle,inner sep=1pt,rounded corners=2mm,below]
(Ainit) {$\begin{array}{l} \ell_\Box\\ \ell_\lozenge \end{array}$};

\draw (0,0) node [rectangle,fill=black!10!white,inner sep=1pt,rounded
  corners=2mm,below] (A) {$\begin{array}{l}
    \{0\}\\ \phantom{a} \end{array}$};

\draw (2,0) node [rectangle,fill=black!10!white,inner sep=1pt,rounded
  corners=2mm,below] (B) {$\begin{array}{l}
    \{0.1\}\\ \{0\} \end{array}$};

\draw (4.5,0) node [rectangle,fill=black!10!white,inner sep=1pt,rounded
  corners=2mm,below] (C) {$\begin{array}{l}
    \{0.2\}\\ \left[0,0.1\right] \end{array}$};

\draw (7.5,0) node [rectangle,fill=black!10!white,inner
  sep=1pt,rounded corners=2mm,below] (D) {$\begin{array}{l}
    \{1.9\}\\ \{0\},\left[1.7,1.8\right] \end{array}$};

\draw (10,0) node [rectangle,fill=black!10!white,inner sep=1pt,rounded
  corners=2mm,below] (E) {$\begin{array}{l}
    \{2\}\\ \{0.1,\} \end{array}$};

\draw (12,0) node [rectangle,fill=black!10!white,inner sep=1pt,rounded
  corners=2mm,below] (F) {$\begin{array}{l}
    \{3\}\\ \phantom{a} \end{array}$};

\draw [-latex'] (A) -- (B) node [pos=0.5,above] {$0.1,a$};
\draw [-latex'] (B) -- (C) node [pos=0.5,above] {$0.1,a$};
\draw [-latex'] (C) -- (D) node [pos=0.5,above] {$1.7,a$};
\draw [-latex'] (D) -- (E) node [pos=0.5,above] {$0.1,b$};
\draw [-latex'] (E) -- (F) node [pos=0.5,above] {$1,b$};
\end{scope}
\end{tikzpicture}
\caption{Several \ATA runs.}
\label{reprun}
\end{figure}
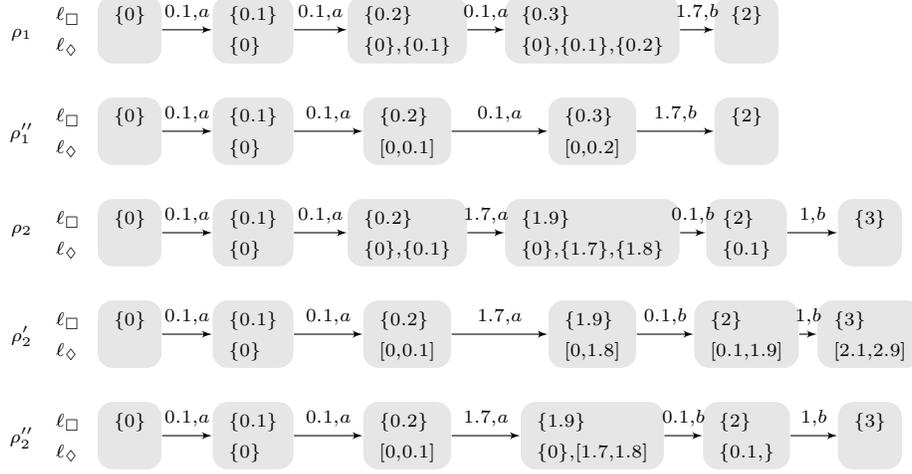

\end{example}

\paragraph{The approximation functions $\appfunc{\Phi}$.} Let us now
formally define the family of \emph{bounded} approximation functions
that will form the basis of our translation to timed automata. We
first give an \emph{upper bound} $M(\Phi)$ on the number of clock
copies (intervals) we need to consider in the configurations to
recognise an MITL formula $\Phi$.  The precise definition of the bound
$M(\Phi)$ is technical and is given by induction on the structure of
the formula. It can be found in the appendix. However, for all MITL
formula $\Phi$: 
$$
M(\Phi)\leq |\Phi| \times \max_{I\in \Ii_\Phi} \left( 4 \times \left\{\left\lceil
  \frac{\inf(I)}{\vert I \vert} \right\rceil\right\} +2, 2 \times \left\{\left\lceil
  \frac{\sup(I)}{\vert I \vert} \right\rceil\right\} +2 \right)
$$ where $\Ii_\Phi$ is the set of all the intervals that occur in
  $\Phi$\footnote{The first component of the maximum comes from $U$
    and the second from $\tilde{U}$.}.

Equipped with this bound, we can define the $\appfunc{\Phi}$
function. Throughout this description, we assume an \ATA $\Aa$ with
set of locations $L$. Let
$S=\{(\ell,I_0),(\ell,I_1),\break\ldots,(\ell,I_m)\}$ be a set of
states of $\Aa$, all in the same location $\ell$, with, as usual $I_0
< I_1 < \cdots < I_m$. Then, we let $\mer{S}=\{(\ell, [0,sup(I_1)]),
(\ell, I_2),\ldots, (\ell,I_m)\}$ if $I_0=[0,0]$ and $\mer{S}=S$
otherwise, i.e., $\mer{S}$ is obtained from $S$ by \emph{grouping}
$I_0$ and $I_1$ iff $I_0=[0,0]$, otherwise $\mer{S}$ does not modify
$S$. Observe that, in the former case, if $I_{1}$ is not a singleton,
then $\nbclocks{\mer{S}}=\nbclocks{S}-1$. Now, we can lift the
definition of \mername to configurations. Let $C$ be a configuration
of $\Aa$ and let $k\in \N$. We let:
\begin{align*}
  \mer{C,k} &= \big\{C'\mid \nbclocks{C'}\leq k\text{ and } \forall
  \ell\in L: C'(\ell)\in\{\mer{C(\ell)}, C(\ell)\}\big\}
\end{align*}
Observe that $\mer{C,k}$ is a (possibly empty) \emph{set of
  configurations}, where each configuration $(i)$ has at most $k$
clock copies, and $(ii)$ can be obtained by applying (if possible) or
not the $\mername$ function to each $C(\ell)$. Let us now define a
family of $k$-bounded approximation functions, based on
$\mername$. Let $k\geq 2\times |L|$ be a bound and let $C$ be a
configuration, assuming that $C(\ell)=\{I^\ell_1,\ldots,
I^\ell_{m_\ell}\}$ for all $\ell\in L$. Then:
\begin{align*}
  F^k(C) &=
  \left\{
    \begin{array}{ll}
     \mer{C,k} &\text{If } \mer{C,k}\neq\emptyset\\
     \big\{(\ell, [inf(I^\ell_1),sup(I^\ell_{m_\ell})])\mid \ell\in L\big\} &\text{otherwise}
    \end{array}
  \right.
\end{align*}
Roughly speaking, the $F^k(C)$ function tries to obtain configurations
$C'$ that approximate $C$ and s.t. $\nbclocks{C'}\leq k$, using the
$\mername$ function. If it fails to, i.e., when $\mer{C,k}=\emptyset$,
$F^k(C)$ returns a single configuration, obtained from $C$ be grouping
all the intervals in each location. The latter case occurs in the
definition of $F^k$ for the sake of completeness. When the \ATA $\Aa$
has been obtained from an MITL formula $\Phi$, and for $k$ big enough
(see hereunder) each $\theta\in\sem{\Phi}$ will be recognised by at
least one $F^k$-run of $\Aa$ that traverses only configurations
obtained thanks to $\mername$.  We can now finally define
$\appfunc{\Phi}$ for all MITL formula $\Phi$, by letting
$\appfunc{\Phi}=F^K$, where $K=\max\{2\times |L|,M(\Phi)\}$. It is
easy to observe that $\appfunc{\Phi}$ is indeed a \emph{bounded
  approximation function}. Then, we can show that, for all MITL
formula $\Phi$, the $\appfunc{\Phi}$-semantics of $\Aa_\Phi$ accepts
exactly $\sem{\Phi}$. To obtain this result, we rely on the following
proposition\footnote{Stated here for the $U$ modality, a similar
  proposition holds for $\tilde{U}$.}:
\begin{proposition}
  \label{PropUntil}
  Let $\Phi$ be an MITL formula, let $K$ be a set of index and,
  $\forall k \in K$, let $\Phi_{k} = \varphi_{1, k}U_{I_{k}}
  \varphi_{2, k}$ be subformulas of $\Phi$. For all $k\in K$, let
  $\ell_{\Phi_{k}}$ be their associated locations in $\Aa_\Phi$.  Let
  $\theta = (\bar{\sigma},\bar{\tau})$ be a timed word and let $J_{k}
  \in \mathcal{I}(\R^{+})$ be closed intervals. The automaton
  $\mathcal{A}_{\Phi}$ $Id$-accepts $\theta$
  from configuration $\lbrace (\ell_{\Phi_{k}},J_{k})_{k \in K}
  \rbrace \text{ iff } \forall k \in K, \exists m_{k} \geq 1 :
  (\theta,m_{k}) \models \varphi_{2} \wedge \tau_{m_{k}} \in
  I_{k}-\inf(J_{k}) \wedge \tau_{m_{k}} \in I_{k}-\sup(J_{k}) \wedge
  \forall 1 \leq m'_{k} < m_{k} : (\theta,m'_{k}) \models
  \varphi_{1}$.
\end{proposition}
To illustrate this proposition, let us consider $\Phi_2 \equiv \top
U_{[2,3]} b$, the associated automaton $\mathcal{A}_{\Phi_2}$ and the
timed word $\theta=(a,0)(b,1)(b,2)$.  Assume that
$\mathcal{A}_{\Phi_2}$ is in configuration $C = \lbrace
(\ell_{\Phi_2}, [0,2]) \rbrace$ and we must read $\theta$ from $C$.
Observe that for all value $y\in [0,2]$, there is a position $m$ in
$\theta$ s.t. $\tau_m \in [2,3]-y$ (i.e., $\tau_m$ satisfies the
temporal constraint of the modality), $(\theta, m) \models b$ and all
intermediate positions satisfy $\top$: $\forall 1 \leq m' < m,
(\theta, m') \models \top$. In other words, $\forall y \in [0,2],
(\theta, 1) \models \top U_{[2,3]-y} b$. Yet, the conditions of the
propositions are not satisfied and, indeed, there is no accepting
$Id$-run of $\mathcal{A}_{\Phi_2}$ from $C$. Indeed, after reading the
first (resp. second) $b$, the resulting configuration contains
$(\ell_{\Phi_2},[1,3])$ (resp. $(\ell_{\Phi_2},[2,4])$). In both
cases, the interval associated to $\ell_{\Phi_2}$ does not satisfy the
clock constraint $x \in [2,3]$ of the transition $x.\delta(\top,b)
\wedge x \in [2,3]$ of $\mathcal{A}_{\Phi_2}$ that enables to leave
location $\ell_{\Phi_2}$. This example also shows that
Proposition~\ref{PropUntil} cannot be obtained as a corollary of the
results by Ouaknine and Worrell \cite{OW07} and deserves a dedicated
proof as part of our contribution.

The property given by Proposition~\ref{PropUntil} is thus crucial to
determine, given an accepting run, whether we can \emph{group} several
intervals and retain an accepting run or not. This observation will be
central to the proof of our main theorem:
\begin{theorem}
  \label{ThmBorne}
  For all MITL formula $\Phi$, $\appfunc{\Phi}$ is a \emph{bounded
    approximation function} and
  $L_{\appfunc{\Phi}}(\mathcal{A}_{\Phi}) =
  L(\mathcal{A}_{\Phi})=\sem{\Phi}$.
\end{theorem}
\begin{proof}[sketch]
  By definition, $\appfunc{\Phi}$ is a \emph{bounded approximation
    function}. Hence, by Proposition~\ref{inclu},
  $L_{\appfunc{\Phi}}(\mathcal{A}_{\Phi}) \subseteq
  L(\mathcal{A}_{\Phi})$. Let $\theta = (\bar{\sigma},\bar{\tau})$ be
  a timed word in $L(\mathcal{A}_{\Phi})$, and let us show that
  $\theta\in L_{\appfunc{\Phi}}(\mathcal{A}_{\Phi})$, by building an
  accepting $\appfunc{\Phi}$-run $\rho'$ on $\theta$. Since, $\theta\in
  L(A_{\Phi})$, there is an accepting $Id$-run $\rho$ of
  $\mathcal{A}_{\Phi}$ on $\theta$. We assume that $\rho
  =C_{0}\xrightarrow{\tau_1,\sigma_1}_{Id} C_{1}
  \xrightarrow{\tau_2-\tau_1,\sigma_2}_{Id} C_{2}\cdots
  \xrightarrow{\tau_n-\tau_{n-1},\sigma_n} C_{n}$.

  We build, by induction on the length of $\rho$, a sequence of runs
  $\rho_0, \rho_1,\ldots,\rho_n$ s.t. for all $0\leq j\leq n$,
  $\rho_{j}=D_0^j\xrightarrow{\tau_1,\sigma_1}_{Id} \cdots
  \xrightarrow{\tau_n-\tau_{n-1},\sigma_n} D_{n}^j$ is an accepting
  run on $\theta$ with the following properties: $(i)$ for all $0\leq
  k\leq j$: $\nbclocks{D_k^j}\leq 4 \times \left\lceil
    \frac{\inf(I_{i})}{\vert I_{i} \vert} \right\rceil +2$, and $(ii)$
  assuming $\tau_0=0$:
  $D_{j}^j\xrightarrow{\tau_{j+1}-\tau_{j},\sigma_{j}}_{Id}\cdots
  \xrightarrow{\tau_n-\tau_{n-1},\sigma_n}_{Id} D_{n}^j$ is an accepting
  $Id$-run on
  $\theta^{j+1}=(\sigma_{j+1}\sigma_{j+2}\ldots\sigma_n,\tau')$, where
  $\tau'=(\tau_{j+1}-\tau_j)(\tau_{j+2}-\tau_{j})\ldots(\tau_n-\tau_{j})$,
  assuming $\tau_{-1}=0$, i.e., $\theta^j$ is the suffix of length
  $n-j$ of $\theta$, where all the timed stamps have been decreased by
  $\tau_{j}$. Clearly, letting $\rho_0=\rho$ satisfied these
  properties. We build $\rho_{k+1}$ from $\rho_k$, by first letting
  $D^{k+1}_0,D^{k+1}_1,\ldots,D^{k+1}_k=D^{k}_0,D^{k}_1,\ldots,D^{k}_k$,
  and then showing how to build $D^{k+1}_{k+1}$ from $D^k_{k+1}$ by
  merging intervals. Let $\ell\in L$. We use the criterion given by
  Proposition~\ref{PropUntil} to decide when to group intervals in
  $D^k_{k+1}(\ell)$. Assume $D^{k}_{k+1}(\ell)=\{J_1, J_2,\ldots,
  J_m\}$. Then:
\begin{itemize}
\item If $D^k_{k+1}(\ell)$ is either empty, or a singleton, we
  let $D^{k+1}_{k+1}(\ell) = D^k_{k+1}(\ell)$.
\item Else, if $J_1\neq [0,0]$, then the reading of $\sigma_{k+1}$ has
  not created a new copy in $\ell$ and we let $D^{k+1}_{k+1}(\ell) =
  D^k_{k+1}(\ell)$ too.
\item Else, $J_1 = [0,0]$ and we must decide whether we group this
  clock copy with $J_2$ or not. Assume\footnote{when $\ell$
    corresponds to the sub-formula $\varphi_1 \tilde{U}_I \varphi_2$,
    we use the proposition similar to Proposition \ref{PropUntil} for
    $\tilde{U}$ to decide whether we group $J_1$ with $J_2$ or not.}
  $\ell$ corresponds to the sub-formula $\varphi_1 U_I
  \varphi_2$. Then:
  \begin{enumerate}
  \item if $\exists m \geq 1$ such that :
    $(\theta^{k+1},\tau^{k+1}_{m}) \vDash \varphi_{2} \wedge
    \tau^{k+1}_{m} \in I - \sup(J_{2}) \wedge \tau^{k+1}_{m}
    \in I \wedge \forall 1 \leq m' < m :
    (\theta^{k+1},\tau^{k+1}_{m'}) \models \varphi_{1}$, then, we
    let $D^{k+1}_{k+1}(\ell) = \mer{D^k_{k+1}(\ell)}$,
\item else, we let $D^{k+1}_{k+1}(\ell) = D^k_{k+1}(\ell)$.
\end{enumerate}
\end{itemize}
We finish the construction of $\rho_{k+1}$ by firing, from
$D^{k+1}_{k+1}$ the same arcs as in the $D^{k}_{k+1} D^k_{k+1}\ldots
D^k_n$ suffix of $\rho_k$, using the
$Id$-semantics. Proposition~\ref{PropUntil} guarantees that we have
grouped the intervals in such a way that this suffix is an
$Id$-accepting run on $\theta^{k+1}$.  Finally, we let $\rho'=\rho_n$
which is an accepting run on $\theta$. We finish the proof by a
technical discussion showing that $\rho_n$ is an
$\appfunc{\Phi}$-run.
\end{proof}

\paragraph{From \ATA to timed automata.} Let us show how we can now
translate $\Aa_\Phi$ into a \emph{timed automaton} that accepts
$\sem{\Phi}$. The crucial point is to define a \emph{bound}, $M(\Phi)$,
on the number of clocks that are necessary to recognise models of
$\Phi$.

A timed automaton (TA) is a tuple $\Bb=(\Sigma, L, \ell_0, X, F,
\delta)$, where $\Sigma$ is a finite \emph{alphabet}, $L$ is finite
set of \emph{locations}, $\ell_0\in L$ is the \emph{initial location},
$X$ is a finite set of \emph{clocks}, $F\subseteq L$ is the set of
\emph{accepting locations}, and $\delta\subseteq
L\times\Sigma\times\mathcal{G}(X)\times 2^X\times L$ is a finite set
of \emph{transitions}, where $\mathcal{G}(X)$ denotes the set of
\emph{guards on $X$}, i.e. the set of all finite conjunctions of
\emph{clock constraints} on clocks from $X$. For a transition $(\ell,
\sigma, g, r, \ell')$, we say that $g$ is its \emph{guard}, and $r$
its \emph{reset}. A \emph{configuration} of a TA is a pair $(\ell,
v)$, where $v:X\mapsto\R^+$ is a \emph{valuation} of the clocks in
$X$. We denote by $\configs{\Bb}$ the set of all configurations of
$\Bb$, and we say that $(\ell,v)$ is \emph{accepting} iff $\ell\in
F$. For all $t\in \R^+$, we have (time successor)
$(\ell,v)\timestep{t}(\ell',v')$ iff $\ell=\ell'$ and $v'=v+t$ where
$v+t$ is the valuation s.t. for all $x\in X$: $(v+t)(x)=v(x)+t$. For
all $\sigma\in \Sigma$, we have (discrete successor)
$(\ell,v)\xrightarrow{\sigma}(\ell',v')$ iff there is $(\ell, \sigma,
g, r, \ell')\in \delta$ s.t. $v\models g$, for all $x\in r$: $v'(x)=0$
and for all $x\in X\setminus r$: $v'(x)=v(x)$. We write
$(\ell,v)\xrightarrow{t,\sigma}(\ell',v')$ iff there is
$(\ell'',v'')\in \configs{\Bb}$
s.t. $(\ell,v)\timestep{t}(\ell'',v'')\xrightarrow{\sigma}(\ell',v')$.
A timed word $\theta=(\bar{\sigma},\bar{\tau})$ with
$\bar{\sigma}=\sigma_1\sigma_2\cdots\sigma_n$ and
$\bar{\tau}=\tau_1\tau_2\cdots\tau_n$ is \emph{accepted} by $\Bb$ iff
there is a \emph{run} of $\Bb$ on $\theta$, i.e. a sequence of
configurations $(\ell_1,v_1)$,\ldots, $(\ell_n,v_n)$ s.t. for all
$0\leq i\leq n-1$:
$(\ell_i,v_i)\xrightarrow{\tau_i-\tau_{i-1},\sigma_i}(\ell_{i+1},
v_{i+1})$, where $v_0$ assigns $0$ to all clocks, and assuming that
$\tau_{-1}$ denotes $0$. We denote by $L(\Bb)$ the \emph{language} of
$\Bb$, i.e. the set of words accepted by $\Bb$.

We can now sketch the translation, the full details can be found in
Appendix~\ref{sec:towards-timed-autom}. Let $\Phi$ be an MITL formula,
and assume $\Aa_\Phi=(\Sigma, L^\Phi, \ell_0^\Phi,
F^\Phi,\delta^\Phi)$. Let us show how to build the TA
$\Bb_\Phi=(\Sigma, L,\ell_0, X, F, \delta)$
s.t. $L(\Bb_\Phi)=L_{\appfunc{\Phi}}(\Aa_\Phi)$. The TA $\Bb_\Phi$ is
built as follows. For a set of clocks $X$, we let $\loc(X)$ be the set
of functions $S$ that associate with each $\ell\in L^\Phi$ a finite
sequence $(x_1,y_1),\ldots,(x_n,y_n)$ of pairs of clocks from $X$,
s.t. each clock occurs only once in all the $S(\ell)$. Then,
$L=\loc(X)$. Observe that $L$ is indeed a finite set. Intuitively, a
configuration $(S, v)$ of $\Bb_\Phi$ encodes the configuration $C$ of
$\Aa_\Phi$ s.t. for all $\ell\in L^\Phi$: $C(\ell)=\{[v(x), v(y)]\mid
(x,y)\in S(\ell)\}$. The other components of $\Bb_\Phi$ are defined as
follows. $\ell_0$ is s.t. $\ell_0(\ell_0^\Phi)=(x,y)$, where $x$ and
$y$ are two clocks arbitrarily chosen from $X$, and
$\ell_0(\ell)=\emptyset$ for all $\ell\in L^\Phi\setminus
\{\ell_0^\Phi\}$. $X$ is a set of clocks s.t. $|X|=M(\Phi)$. $F$ is
the set of all locations $S$ s.t. $\{\ell\mid
S(\ell)\neq\emptyset\}\subseteq F^\Phi$. Finally, $\delta$ allows
$\Bb_\Phi$ to simulate the $\appfunc{\Phi}$-semantics of $\Aa_\Phi$ :
the non determinism of $\delta$ enables to guess which clocks must be
grouped to form appropriate intervals (see appendix for all the
details).

\begin{theorem}\label{theo:TA-from-MITL}
  For all MITL formula $\Phi$, $\Bb_\Phi$ has $M(\Phi)$ clocks and $O(
  (\vert \Phi \vert)^{(m . \vert \Phi \vert)})$ locations, where $m =
  \max_{I\in\Ii_\Phi} \left\{ 2 \times \left\lceil
      \frac{\inf(I)}{\vert I \vert} \right\rceil + 1, \left\lceil
      \frac{\sup(I)}{\vert I \vert} \right\rceil + 1 \right\}$.
\end{theorem}

\section{Future works: towards efficient MITL model checking\label{sec:future-works:-toward}}
Let us close this work by several observations that could yield
efficient model checking algorithm for MITL. Let $\Cc$ be a timed
automaton, and let $\Phi$ be an MITL formula. Obviously, one can
perform \emph{automaton-based model checking} by computing a TA
$\Bb_{\neg\Phi}$ accepting $\sem{\neg\Phi}$ (using the technique
presented in Section~\ref{sec:from-mitl-timed}, or the technique
of~\cite{AFH96}), and explore their synchronous product
$\Cc\times\Bb_{\neg\Phi}$ using classical region-based or zone-based
techniques \cite{AD94}.  This approach has an important drawback in
practice: the number $M$ of clocks of the $\Bb_{\neg\Phi}$ TA is
usually very high (using our approach or the~\cite{AFH96} approach),
and the algorithm exploring $\Cc\times\Bb_{\neg\Phi}$ will have to
maintain data structures (regions or zone) ranging over $N+M$ clocks,
where $N$ is the number of clocks of $\Cc$.

A way to avoid this blow up in the number of clocks is to perform the
model-checking using the \ATA $\Aa_{\neg\Phi}$ (using its
$\appfunc{\neg\Phi}$ semantics) instead of the TA
$\Bb_{\neg\Phi}$. First, the size of $\Aa_{\neg\Phi}$ is \emph{linear}
in the size of $\Phi$, and is straightforward to build. Second, a
configuration of $\Cc\times\Aa_{\neg\Phi}$ stores only the clocks that
correspond to \emph{active copies} of $\Aa_{\neg\Phi}$, which, in
practice, can be much smaller than the number of clocks of
$\Bb_{\neg\Phi}$. Third, this approach allows to retain the structure
of the \ATA in the transition system of $\Cc\times\Aa_{\neg\Phi}$,
which allows to define \emph{antichain based algorithms} \cite{DR10},
that rely on a \emph{partial order on the state space} to detect
redundant states and avoid exploring them. Such an approach, has been
applied in the case of LTL model-checking \cite{DDMR08}. It relies
crucially on the translation of LTL formulas to \emph{alternating
  automata}, and yields dramatic improvements in the practical
performance of the algorithm.

To obtain such algorithms, we need a \emph{symbolic data structure} to
encode the configurations of $\Cc\times\Aa_{\neg\Phi}$. Such a data
structure can be achieved by lifting, to our \emph{interval
  semantics}, the technique from \cite{OW05} that consists in encoding
\emph{regions} of \ATA configurations by means of \emph{finite
  words}. Remark that this encoding differs from the classical regions
for TA \cite{AD94}, in the sense that the \emph{word} encoding allows
the number of clocks to change along paths of the transition system.

These ideas explain what we believe are the \emph{benefits} of using
an \ATA based characterisation of MITL formulas. The precise
definition of the model checking algorithm sketched above is the topic
of our current research.


\clearpage

\appendix

\section{Proof of Proposition \ref{inclu}}

Before proving Proposition~\ref{inclu}, we make several useful
observations about the transition relation of an \ATA. Let $\delta$ be
the transition function of some \ATA, let $\ell$ be a location, let
$\sigma$ be a letter, and assume $\delta(\ell,\sigma)=\bigvee_j a_j$,
where each $a_j$ is an \emph{arc}, i.e. a conjunction of atoms of the
form: $\ell'$, $x.\ell'$, $x\bowtie c$, $0\bowtie c$, $\top$ or
$\bot$. Then, we observe that \emph{each minimal model} of
$\delta(\ell,\sigma)$ wrt some interval $I$ corresponds to
\emph{firing one of the arcs $a_j$ from $(\ell,I)$}. That, each
minimal model can be obtained by choosing an arc $a_j$ from
$\delta(\ell,\sigma)$, and applying the following procedure. Assume
$a_j= \ell_1\wedge\cdots\wedge \ell_n\wedge
x.(\ell_{n+1}\wedge\cdots\wedge \ell_{m})\wedge \varphi$, where
$\varphi$ is a conjunction of clock constraints. Then, $a_j$ is
firable from a minimal model $(\ell,\sigma)$ iff $I\models \varphi$
(otherwise, no minimal model can be obtained from $a_j$). In this
case, the minimal model is $\{(\ell_i,I)\mid 1\leq i\leq n\}\cup
\{(\ell_i,[0,0])\mid n+1\leq i\leq m\}$. This generalises to
configurations $C=\{(\ell_1,I_1),\ldots, (\ell_n,I_n)\}$: successors
$C'$ are obtained by selecting a firable arc from each $(\ell_i,
I_i)$, and taking the union of the resulting configurations.

\noindent \textbf{Proposition \ref{inclu}.} \textit{For all \ATA
  $\Aa$, for all $f \in \appf{\Aa}$: $L_{f}(\Aa) \subseteq L(\Aa)$.}
\begin{proof}
  Let us consider a timed word $\theta = (\bar{\sigma}, \bar{\tau})$
  in $L_{f}(\mathcal{A})$ and let us show that $\theta\in L(\Aa)$. Let
  us assume that $\leng{\theta} = n$. Let
  $\rho=\lbrace(\ell_{0},[0,0])\rbrace
  \overset{t_{1}}{\rightsquigarrow} C_{1}
  \overset{\sigma_{1}}{\longrightarrow_{f}} C_{2}
  \overset{t_{2}}{\rightsquigarrow} C_{3}
  \overset{\sigma_{2}}{\longrightarrow_{f}} \dots
  \overset{t_{n}}{\rightsquigarrow} C_{2n-1}
  \overset{\sigma_{n}}{\longrightarrow_{f}} C_{2n}$, for
  $\lbrace(\ell_{0},[0,0])\rbrace = C_{0}$ be the accepting $f$-run of
  $\mathcal{A}$ on $\theta$, with $\lbrace(\ell_{0},[0,0])\rbrace =
  C_{0}$ and where $C_{2n}$ is accepting.  Let us build, from $\rho$,
  an accepting $Id$-run $\rho'=D_{0} \overset{t_{1}}{\rightsquigarrow}
  D_{1} \overset{\sigma_{1}}{\longrightarrow_{Id}} D_{2}
  \overset{t_{2}}{\rightsquigarrow} D_{3}
  \overset{\sigma_{2}}{\longrightarrow_{Id}} \dots
  \overset{t_{n}}{\rightsquigarrow} D_{2n-1}
  \overset{\sigma_{n}}{\longrightarrow_{Id}} D_{2n}$ s.t.:
  \begin{enumerate}
  \item $D_{0} = \lbrace(\ell_{0}, [0,0])\rbrace$,
  \item $D_{2n}$ is accepting, and
  \item for all $1\leq i\leq n$, for all $(\ell,[v,v]) \in D_{2i}$
    there is $(\ell,I) \in C_{2i}$ s.t. $v \in I$
  \end{enumerate}

  We build $\rho'$ by induction on the positions along $\rho$:
  
  \textbf{\underline{Basis}} : $(i=0)$ Since $D_0=C_0$, the property
  holds trivially.
  
  \textbf{\underline{Induction}} : $(i=k+1)$ The induction hypothesis
  is that we have built the prefix of $\rho'$ up to $D_{2k}$, and
  that, for all $j\leq k$, for all $(\ell,[v,v]) \in D_{2j}$ there is
  $(\ell,I) \in C_{2j}$ s.t. $v \in I$. Let us show how to build
  $D_{2k+1}$ and $D_{2(k+1)}$.

  \begin{itemize}
  \item We first take care of the time transition. Let $D_{2k+1} =
    D_{2k} + t_{k+1}$. Clearly, $D_{2k}
    \overset{t_{k+1}}{\rightsquigarrow} D_{2k+1}$.  In $\rho$, we have
    $C_{2k} \overset{t_{k+1}}{\rightsquigarrow} C_{2k+1}$, and so
    $C_{2k+1} = C_{2k} + t_{k+1}$.  It is straightforward to prove
    that this maintains the induction hypothesis:
    \begin{eqnarray}
      \forall(\ell,[v,v]) \in D_{2k+1}: \exists(\ell,I) \in
      C_{2k+1} : v \in I\label{eq:1}
    \end{eqnarray}  
  \item We must build $D_{2k+2}$ corresponding to the transition
    $D_{2k+1} \xrightarrow{\sigma_{k+1}}_{Id} D_{2k+2}$.  Let us
    assume that $D_{2k+1}=\{(\ell_1,[v_1,v_1]),\ldots (\ell_p,
    [v_p,v_p])\}$ for some $p$, and let us assume that $C_{2k+1} =
    \{(\ell'_1,I_1),\ldots, (\ell'_q,I_q)\}$ for some $q$.  Let
    $h:\{1,\ldots,p\}\mapsto\{1,\ldots,q\}$ be a function s.t. for all
    $1\leq j\leq p$: $v_j\in I_{h(j)}$. This function exists,
    by~(\ref{eq:1}), and we will rely on it to build $D_{2k+2}$ from
    $C_{2k+2}$.

    In $\rho$, $C_{2k+1} \xrightarrow{\sigma_{k+1}}_f C_{2k+2}$. Then,
    $C_{2k+2}\in f(E)$, where $E=\cup_{1\leq j\leq q}E_j$ and each
    $E_j$ is a minimal model of $\delta(\ell_{j}',\sigma_{k+1})$ with
    respect to $I_{j}$.  Each of those minimal models corresponds to
    an arc starting from $\ell_j'$, let us denote this arc by
    $a_j$. Remark that, for all $j$, $I_j$ satisfies the guard of
    $a_j$, since $\rho$ is a genuine run.

    Then, we let $D_{2k+2}$ be the configuration obtained by taking,
    for all $1\leq j\leq p$, the arc $a_{h(j)}$ from $(\ell_j,
    [v_j,v_j])\in D_{2k+1}$. Clearly, $v_j$ satisfies the guard $g$ of
    $a_{h(j)}$, because $v_j\in I_{h(j)}$, by definition of $h$, and
    $I_{h(j)}$ satisfies $g$. It is also easy to check that
    $D_{2k+2}=\cup_{1\leq j\leq p} M_j$, where each $M_j$ is a minimal
    model of $\delta(\ell_j, \sigma_{k+1})$ wrt. $[v_j,v_j]$. Hence,
    $D_{2k+1}\longrightarrow_{Id}D_{2k+2}$. Finally, since we fire the
    same arcs in both $\rho$ and $\rho'$, the resets are the same in
    both cases. By definition of the approximation function, we
    conclude that $C_{2k+1}$ and $D_{2k+2}$ satisfy the induction
    hypothesis.
\end{itemize}
This induction builds the run $\rho'$ and shows that for all
$(\ell,[v,v]) \in D_{2n}$, there is $(\ell,I) \in C_{2n}$ s.t. $v \in
I $.  As $\rho$ is accepting, $C_{2n}$ is an accepting configuration
and all the states it contains are accepting, i.e. $\forall (\ell,I)
\in C_{2n}, \ell$ is an accepting location.  We deduce from this that
$D_{2n}$ is an accepting configuration.
\end{proof}

\section{Proof of Theorem \ref{ThmBorne}}

Before proving Theorem~\ref{ThmBorne}, we give a precise
characterisation of the bound $M(\Phi)$. Let $\Phi$ be an MITL formula
in negative normal form.  We define $M(\Phi)$, thanks to
$M^{\infty}(\Phi)$ and $M^{1}(\Phi)$ defined as follows
\begin{itemize}
\item if $\Phi = \sigma$ or $\Phi=\neg \sigma$ (with $\sigma\in
  \Sigma$), then $M(\Phi) = 2$ and $M^{\infty}(\Phi) = M^{1}(\Phi) =
  0$.
\item if $\Phi = \varphi_{1} \wedge \varphi_{2}$, then $M(\Phi) = \max \{2,
  M^{1}(\varphi_{1}) + M^{1}(\varphi_{2})\}$, $M^{\infty}(\Phi) =
  M^{\infty}(\varphi_{1}) + M^{\infty}(\varphi_{2})$ and $M^{1}(\Phi)
  = M^{1}(\varphi_{1}) + M^{1}(\varphi_{2})$.
\item if $\Phi = \varphi_{1} \vee \varphi_{2}$, then $M(\Phi) = \max
  \{2, M^{1}(\varphi_{1}), M^{1}(\varphi_{2})\}$, $M^{\infty}(\Phi) =\break
  \max \{ M^{\infty}(\varphi_{1}), M^{\infty}(\varphi_{2})\}$ and
  $M^{1}(\Phi) = \max \{ M^{1}(\varphi_{1}), M^{1}(\varphi_{2})\}$.
\item if $\Phi = \varphi_{1}U_{I} \varphi_{2}$, then $M(\Phi) = \max
  \{ 2, M^{\infty}(\varphi_{1}) +M^{1}(\varphi_{2}) +1 \}$,
  $M^{\infty}(\Phi) = \left(4\times\left\lceil \frac{inf(I)}{\vert I
        \vert} \right\rceil +2\right) +M^{\infty}(\varphi_{1})
  +M^{\infty}(\varphi_{2})$ and $M^{1}(\Phi) = M^{\infty}(\varphi_{1})
  +M^{1}(\varphi_{2}) +1$.
\item if $\Phi = \varphi_{1} \tilde{U}_{I} \varphi_{2}$, then $M(\Phi)
  = \max \{2, M_{1}(\varphi_{1}) +M_{\infty}(\varphi_{2}) +1\}$,
  $M_{\infty}(\Phi) = \left(2\times \left\lceil \frac{sup(I)}{\vert I
        \vert} \right\rceil +2\right) +M_{\infty}(\varphi_{1})
  +M_{\infty}(\varphi_{2})$ and $M_{1}(\Phi) = M_{1}(\varphi_{1})
  +M_{\infty}(\varphi_{2}) +1$.
\end{itemize}

Then, let us recall useful results from \cite{OW05,OW07} that enable
to prove Proposition~\ref{PropUntil}. 

\begin{proposition}[\cite{OW05}]
\label{PropOW}
Let $\Phi = \varphi_{1}U_{I} \varphi_{2}$ or
$\Phi=\varphi_{1}\tilde{U}_{I} \varphi_{2}$ be an MITL formula and
$\ell_{\Phi}$ the associated location in $\Aa_\Phi$.  Let $\theta$ be
a timed word. The automaton $\Aa_\Phi$ $Id$-accepts $\theta$
from configuration $\lbrace(\ell_{\Phi},0)\rbrace \text{ iff }
(\theta,1) \models \Phi$.
\end{proposition}

The following corollary directly follows from
Proposition~\ref{PropOW}. 

\begin{corollary}
\label{CorOW}
Let $\Phi$ be an MITL formula, let $K$ be a set of index and, $\forall
k \in K$ let $\Phi_{k} = \varphi_{1, k}U_{I_{k}} \varphi_{2, k}$ or
$\Phi_k\varphi_{1, k}\tilde{U}_{I_{k}} \varphi_{2, k}$ be subformulas
of $\Phi$. For all $k\in K$, let $\ell_{\Phi_{k}}$ the associated
locations in $\Aa_\Phi$.  Let $\theta$ be a timed word and, for all
$k\in K$. Then, $\mathcal{A}_\Phi$ $Id$-accepts $\theta$ from
configuration $\lbrace(\ell_{\Phi_{k}},0)_{k \in K} \rbrace \text{ iff
} (\theta,1) \models \underset{k \in K}{\bigwedge} \varphi_{1, k}
U_{I_{k}} \varphi_{2, k}$.
\end{corollary}

Let us now adapt this result to the cases where the automaton reads
the word from states of the form $(\ell, v)$, with $v$ potentially
$\neq 0$:

\begin{lemma}
  \label{LemSingleton}
  Let $\Phi$ be an MITL formula, let $K$ be a set of index and,
  $\forall k \in K$ let $\Phi_{k} = \varphi_{1, k}U_{I_{k}}
  \varphi_{2, k}$ or $\Phi_{k} = \varphi_{1, k}\tilde{U}_{I_{k}}
  \varphi_{2, k}$ be sub-formulas of $\Phi$. For all $k\in K$, let
  $\ell_{\Phi_{k}}$ be their associated locations in $\Aa_\Phi$.  Let
  $\theta$ be a timed word and $v_{k} \in \R^{+}$ $(\forall k \in
  K)$. The automaton $\mathcal{A}_\Phi$ $Id$-accepts $\theta$ from
  configuration $\lbrace(\ell_{\Phi_{k}},v_k)_{k \in K} \rbrace \text{
    iff } (\theta,1) \models \underset{k \in K}{\bigwedge}
  (\Phi_k-v_k)$, where, for all $k\in K$, $\Phi_k-v_k$ denotes the
  formula obtained from $\Phi_k$ by replacing the $I_k$ interval on
  the modality by $I_k-v_k$.
\end{lemma}
\begin{proof}
  We prove that, for all $k\in K$ s.t. the outer modality of $\Phi_k$
  is $U$: $\Aa_\Phi$ $Id$-accept $\theta$ from $(\ell_{\Phi_{k}},v_k)$
  iff $\theta\models\varphi_{1, k} U_{I_{k}-v_{k}} \varphi_{2,
    k}$. The same arguments adapt to the $\tilde{U}$ case, and the
  Lemma follows.

  Assume $\theta= (\bar{\sigma},\bar{\tau})$ with
  $\bar{\sigma}=\sigma_1\sigma_2\ldots\sigma_n$ and
  $\bar{\tau}=\tau_1\tau_2\ldots\tau_n$. For all $v\in \R^+$, let
  $\theta+v$ denote the timed word $(\bar{\sigma}, \tau')$, where
  $\tau'=(\tau_1+v)(\tau_2+v)\ldots(\tau_n+v)$. Observe that, by
  definition of the semantics of MITL, $\theta\models \varphi_1 U_I
  \varphi_2$ iff $\theta+v\models \varphi_1 U_{I+v}\varphi_2$ (remark
  that the $\varphi_1$ and $\varphi_2$ formulas are preserved).

  First, assume that $\theta\models \varphi_{1, k} U_{I_{k}-v_{k}}
  \varphi_{2, k}$ and let us show that $\Aa_\Phi$ $Id$-accepts
  $\theta$ from $(\ell_{\Phi_{k}},v_k)$.  Since $\theta\models \varphi_{1, k}
  U_{I_{k}-v_{k}} \varphi_{2, k}$, $\theta+v_k\models \varphi_{1,k}
  U_{I_k}\varphi_{2,k} \equiv \Phi_{k}$. Then, by Proposition~\ref{PropOW}, $\Aa_\Phi$
  $Id$-accepts $\theta+v_k$ from $(\ell_{\Phi_{k}},0)$. Let
  $(\ell_\Phi,0)\xrightarrow{\tau_1+v_k,\sigma_1}_{Id} C_1\xrightarrow{\tau_2-\tau_1,\sigma_2}_{Id}
  \cdots\xrightarrow{\tau_n-\tau_{n-1},\sigma_n}_{Id} C_n$ be an accepting
  $Id$-run of $\Aa_\Phi$ from $(\ell_{\Phi_{k}},0)$ on $\theta+v_k$. Obviously, the first time
  step can be decomposed as follows:
  $(\ell_{\Phi_{k}},0)\timestep{v_k}(\ell_{\Phi_{k}},v_k)\xrightarrow{\tau_1,\sigma_1}_{Id} C_1
  \xrightarrow{\tau_2-\tau_1,\sigma_2}_{Id}
  \cdots\xrightarrow{\tau_n-\tau_{n-1},\sigma_n}_{Id} C_n$, where the prefix
  starting in $(\ell_{\Phi_{k}},v_k)$ is an accepting $Id$-run on
  $\theta$. We conclude that $\Aa_\Phi$ $Id$-accepts $\theta$ from
  $(\ell_{\Phi_{k}},v_k)$.

  By using the same arguments, we can prove that $\Aa_\Phi$
  $Id$-accepts $\theta$ from $(\ell_{\Phi_{k}},v_k)$ \emph{implies} that
  $\theta\models \varphi_{1, k} U_{I_{k}-v_{k}} \varphi_{2, k}$.  
\end{proof}

Let us show that Lemma~\ref{LemSingleton} extends to non singular intervals:

\begin{lemma}
\label{LemProp2}
Let $\Phi$ be an MITL formula, let $K$ a set of index and, $\forall k
\in K$, let $\Phi_{k}$ be sub-formulas of $\Phi$ of the form either
$\varphi_{1, k}U_{I_{k}} \varphi_{2, k}$ or $\varphi_{1,
  k}\tilde{U}_{I_{k}} \varphi_{2, k}$. For all $k\in K$, let
$\ell_{\Phi_{k}}$ be their associated locations in $\Aa_\Phi$.  Let
$\theta$ be a timed word and $J_{k} \in \mathcal{I}(\R^{+})$ $(\forall
k \in K)$. The automaton $\mathcal{A}_{\Phi}$ $Id$-accepts $\theta$
from configuration $\lbrace(\ell_{\Phi_{k}},J_{k})_{k \in K} \rbrace
\text{ iff } \forall k \in K, \mathcal{A}_\Phi$ $Id$-accepts $\theta$
from configuration $\lbrace (\ell_{\Phi_{k}},J_{k}) \rbrace$.
\end{lemma}
\begin{proof}
  It is straightforward by definition of runs on $\mathcal{A}_{\Phi}$: the
  time elapsed is reported on each state $(\ell_{\Phi_{k}},J_{k})$ and
  the reading of a letter gives a minimal model for each state
  $(\ell_{\Phi_{k}},J_{k})$ before to merge them into a unique new
  configuration. 
\end{proof}

Now, we recall a result from \cite{OW07}:

\begin{lemma}[\cite{OW07}]
\label{claimOW}
Let $\Phi$ be an MITL formula and $\varphi$ a sub-formula of $\Phi$.
Let $\theta = (\bar{\sigma}, \bar{\tau})$ be a timed word and $\rho$ :
$C_{0} = \lbrace(\ell_{\Phi},0)\rbrace
\overset{t_{1}}{\rightsquigarrow} C_{1}
\overset{\sigma_{1}}{\longrightarrow_{Id}} C_{2}
\overset{t_{2}}{\rightsquigarrow} C_{3}
\overset{\sigma_{2}}{\longrightarrow_{Id}} \dots
\overset{t_{n}}{\rightsquigarrow} C_{2n-1}
\overset{\sigma_{n}}{\longrightarrow_{Id}} C_{2n}$ an accepting
$Id$-run of $\mathcal{A}_{\Phi}$ on $\theta$. $\forall 1 \leq i \leq
n$, if $C_{2i} \models_{[0,0]} \delta(\varphi, \sigma_{i})$, then
$(\theta,i) \models \varphi$.
\end{lemma}

We can now prove Proposition~\ref{PropUntil}.

\noindent \textbf{Proposition \ref{PropUntil} : } \textit{Let $\Phi$ be an
MITL formula, let $K$ be a set of index and, $\forall k \in K$, let
$\Phi_{k} := \varphi_{1, k}U_{I_{k}} \varphi_{2, k}$ be sub-formulas of
$\Phi$. For all $k\in K$, let $\ell_{\Phi_{k}}$ be their associated
locations in $\Aa_\Phi$.  Let $\theta = (\bar{\sigma},\bar{\tau})$ be
a timed word and $J_{k} \in \mathcal{I}(\R^{+})$ closed.
The automaton $\mathcal{A}_{\Phi}$ $Id$-accepts $\theta$ from
configuration $\lbrace (\ell_{\Phi_{k}},J_{k})_{k \in K} \rbrace
\text{ iff } \forall k \in K$: $\exists m_{k} \geq 1$: $(\theta,m_{k})
\models \varphi_{2} \wedge \tau_{m_{k}} \in I_{k}-\inf(J_{k}) \wedge
\tau_{m_{k}} \in I_{k}-\sup(J_{k}) \wedge \forall 1 \leq m'_{k} <
m_{k} : (\theta,m'_{k}) \models \varphi_{1}$.}
\begin{proof}
Thanks to Lemma \ref{LemProp2}, we only need to prove the following.
Let $\Psi$ be an MITL formula, $\Phi := \varphi_{1}U_{I} \varphi_{2}$
a sub-formula of $\Psi$ and $\ell_{\Phi}$ its associated location in
$\Aa_\Psi$.  Let $\theta = (\bar{\sigma},\bar{\tau})$ be a timed word
and $J \in \mathcal{I}(\R^{+})$ closed. The automaton
$\Aa_\Psi$ $Id$-accepts $\theta$ from configuration
$\lbrace (\ell_{\Phi},J) \rbrace \text{ iff } \exists m \geq 1 :
(\theta,m) \models \varphi_{2} \wedge \tau_{m} \in I-\inf(J) \wedge
\tau_{m} \in I-\sup(J) \wedge \forall 1 \leq m' < m : (\theta,m')
\models \varphi_{1}$.

\medskip

\noindent ($\Rightarrow$) As automaton $\Aa_\Psi$
$Id$-accepts $\theta$ from configuration
$\lbrace(\ell_{\Phi},J)\rbrace$, there exists an accepting
$Id$-run of $\mathcal{A}_{\Psi}$ on $\theta$ from
$\lbrace(\ell_{\Phi},J)\rbrace$, say 
\[
\rho = C_{0}  \overset{t_{1}}{\rightsquigarrow} C_{1}
\overset{\sigma_{1}}{\longrightarrow_{Id}} C_{2}
\overset{t_{2}}{\rightsquigarrow} C_{3}
\overset{\sigma_{2}}{\longrightarrow_{Id}} \dots
\overset{t_{n}}{\rightsquigarrow} C_{2n-1}
\overset{\sigma_{n}}{\longrightarrow_{Id}} C_{2n},
\] 
where $C_{0} = \lbrace(\ell_{\Phi},J)\rbrace$ and $C_{2n}$ is
accepting.  For all i, when reading $\sigma_{i}$, two transitions can
be taken: either $x.\delta(\varphi_{2},\sigma) \wedge x \in I$ or
$x.\delta(\varphi_{1},\sigma) \wedge \varphi_{1} U_{I} \varphi_{2}
\wedge x \leq \sup(I)$.  Let $m$ be the first position in the run
where the transition $x.\delta(\varphi_{2},\sigma) \wedge x \in I$ is
taken. Such a position must exist because $\ell_{\Phi}$ is not an
accepting location but $\rho$ is an accepting $Id$-run. Then, for all
$m' < m$, when reading $\sigma_{m'}$, the transition
$x.\delta(\varphi_{1},\sigma_{j}) \wedge \varphi_{1} U_{I} \varphi_{2}
\wedge x \leq \sup(I)$ is taken : it does not reset clock copies that
stay in $\ell_{\Phi}$.  So, the part of configuration $C_{2m-1}$
associated with location $\ell_{\Phi}$ is
$\lbrace(\ell_{\Phi},J+\tau_{m})\rbrace$.  As the transition
$x.\delta(\varphi_{2},\sigma_{m}) \wedge x \in I$ is then taken,
$J+\tau_{m}$ must satisfy $x \in I$, i.e. (by definition of the
minimal model) $\forall v+\tau_{m} \in J+\tau_{m}, v+\tau_{m} \in I$,
i.e. : $\forall v \in J, \tau_{m} \in I-v$ and in particular (as J
closed) $\tau_{m} \in I-\inf(J) \wedge \tau_{m} \in I-\sup(J)$.
Moreover, as $\rho$ is an accepting $Id$-run, the part
$x.\delta(\varphi_{2},\sigma_{m})$ of the transition taken from
$\lbrace(\ell_{\Phi},J+\tau_{m})\rbrace$ corresponds to the fact that
$C_{2m} \models_{[0,0]} \delta(\varphi_{2},\sigma_{m})$, thanks to
Lemma \ref{claimOW}, we know it means that $(\theta,m) \models
\varphi_{2}$. In the same way, with the reading of $\sigma_{m'}$, for
$1 \leq m' < m$, the transition $x.\delta(\varphi_{1},\sigma_{m'})
\wedge \varphi_{1} U_{I} \varphi_{2} \wedge x \leq \sup(I)$ was taken.
As $\rho$ is an accepting $Id$-run, the part
$x.\delta(\varphi_{1},\sigma_{m'})$ of the transition taken from
$\lbrace(\ell_{\Phi},J+\tau_{m})\rbrace$ corresponds to the fact that
$C_{2m'} \models_{[0,0]} \delta(\varphi_{1},\sigma_{m'})$, thanks to
Lemma \ref{claimOW}, we know it means that $(\theta,m') \models
\varphi_{1}$. We conclude that $\exists m \geq 1 : (\theta,m) \models
\varphi_{2} \wedge \tau_{m} \in I-\inf(J) \wedge \tau_{m} \in
I-\sup(J) \wedge \forall 1 \leq m' < m : (\theta,m') \models
\varphi_{1}$.

\medskip

\noindent ($\Leftarrow$) In the sequel, we use the following notation.
Assume that $\theta = (\bar{\sigma},\bar{\tau})$, where $\bar{\sigma}
= \sigma_{1} \sigma_{2} \dots \sigma_{n} \text{ and } \bar{\tau} =
\tau_{1} \tau_{2} \dots \tau_{n}$.  For all $1\leq k \leq n$, we
denote by $\theta_{k} = (\bar{\sigma}_{k},\bar{\tau}_{k})$, where
$\bar{\sigma}_{k} = \sigma_{k} \sigma_{k+1} \dots \sigma_{n} \text{
  and } \bar{\tau}_{k} = \tau'_{1} \tau'_{2} \dots \tau'_{n-k}$ the
timed word such that $\forall 1 \leq i \leq n-k, \tau'_{i} =
\tau_{i+k} - \tau_{k}$.\\
We will
construct an accepting $Id$-run $\rho$ of $\mathcal{A}_{\Psi}$ on
$\theta$ from configuration $\lbrace(\ell_{\Phi},J)\rbrace$, say
$C_{0} = \lbrace(\ell_{\Phi},J)\rbrace
\overset{t_{1}}{\rightsquigarrow} C_{1}
\overset{\sigma_{1}}{\longrightarrow_{Id}} C_{2}
\overset{t_{2}}{\rightsquigarrow} C_{3}
\overset{\sigma_{2}}{\longrightarrow_{Id}} \dots
\overset{t_{n}}{\rightsquigarrow} C_{2n-1}
\overset{\sigma_{n}}{\longrightarrow_{Id}} C_{2n}$, where $C_{2n}$ is
accepting.  By hypothesis, $\exists m \geq 0$ such that (a)
$(\theta,m) \models \varphi_{2}$, (b) $\tau_{m} \in I-\inf(J) \wedge
\tau_{m} \in I-\sup(J)$ and (c) $\forall 1 \leq m' < m : (\theta,m')
\models \varphi_{1}$.  From $\ell_{\Phi}$ we have two possible
transitions $x.\delta(\varphi_{2},\sigma) \wedge x \in I$ and
$x.\delta(\varphi_{1},\sigma) \wedge \varphi_{1} U_{I} \varphi_{2}
\wedge x \leq \sup(I)$.  We construct $\rho$ in way it consists in
following the transition $x.\delta(\varphi_{1},\sigma_{m'}) \wedge
\varphi_{1} U_{I} \varphi_{2} \wedge x \leq \sup(I)$, $\forall 1 \leq
m' < m$, and the transition $x.\delta(\varphi_{2},\sigma_{m}) \wedge x
\in I$ reading $\sigma_{m}$.  We must prove that $\rho$ is an
accepting $Id$-run of $\mathcal{A}_{\Psi}$ on $\theta$.\\ Remark that
following transition $x.\delta(\varphi_{1},\sigma_{m'}) \wedge
\varphi_{1} U_{I} \varphi_{2} \wedge x \leq \sup(I)$, $\forall 1 \leq
m' < m$, in particular, we loop on $\Phi$ without reset of clock.  It
means that, $\forall 1 \leq m' < m$, $C_{2m'+1}(\ell_{\Phi}) = \lbrace
J+\tau_{m'}\rbrace$.  So, it is possible to take transition
$x.\delta(\varphi_{1},\sigma_{m'}) \wedge \varphi_{1} U_{I}
\varphi_{2} \wedge x \leq \sup(I)$ reading $\sigma_{m'}$ from
$\ell_{\Phi}$ because the interval associated to this location is then
$J+\tau_{m'}$ and satisfies the clock constraint $x \leq \sup(I)$ :
(b) implies that $\tau_{m'} < \tau_{m} \leq \sup(I) - \sup(J)$, so
$\sup(J) + \tau_{m'} \leq \sup(I)$, and so $\forall j+\tau_{m'} \in
J+\tau_{m'}, j+\tau_{m'} \leq \sup(I)$.  Moreover, we know that
$\forall 1 \leq m' < m, (\theta,m') \models \varphi_{1}$.  It means
that, $\forall 1 \leq m' < m$, the automaton
$\mathcal{A}_{\varphi_{1}}$ $Id$-accepts $\theta^{m'}$ from
$\lbrace(\varphi_{1, init},0)\rbrace$, i.e. there is an accepting
$Id$-run of $\mathcal{A}_{\varphi_{1}}$ on $\theta^{m'}$ taking
transition $x.\delta(\varphi_{1},\sigma_{m'})$ (the unique transition
we can take from location $\varphi_{1, init}$).  However, the
locations of $\mathcal{A}_{\varphi_{1}}$ in which leads
$x.\delta(\varphi_{1},\sigma_{m'})$ can be assimilated to the
locations of $\mathcal{A}_{\Phi}$ corresponding to the same formulas
(see definitions of such automata and their locations).  So, there is
also an accepting run of $\mathcal{A}_{\Phi}$ on $\theta^{m'}$ taking
transition $x.\delta(\varphi_{1},\sigma_{m'})$ ($\forall 1 \leq m' <
m$).  As transitions $x.\delta(\varphi_{1},\sigma_{m'}) \wedge
\varphi_{1} U_{I} \varphi_{2} \wedge x \leq \sup(I)$ loop on
$\ell_{\Phi}$, when reading $\sigma_{m}$, the interval $J+\tau_{m}$ is
still associated to location $\ell_{\Phi}$.  $\rho$ then consists in
taking transition $x.\delta(\varphi_{2},\sigma_{m}) \wedge x \in
I$. It is possible to take this transition reading $\sigma_{m}$
because $J+\tau_{m}$ satisfies the clock constraint $x \in I$ : as
$\tau_{m} \in I-\inf(J) \wedge \tau_{m} \in I-\sup(J)$ and J is an
interval, $\forall j \in J, \tau_{m} \in I-j$, i.e. $\forall j \in J,
j+\tau_{m} \in I$ and so $\forall v \in J+\tau_{m}, v \in I$.
Moreover, we know that $(\theta,m') \models \varphi_{2}$.  It means
that the automaton $\mathcal{A}_{\varphi_{2}}$ $Id$-accepts
$\theta^{m}$ from $\lbrace(\varphi_{2, init},0)\rbrace$, i.e. there is
an accepting $Id$-run of $\mathcal{A}_{\varphi_{2}}$ on $\theta^{m}$
taking transition $x.\delta(\varphi_{2},\sigma_{m})$ (the unique
transition we can take from location $\varphi_{2, init}$). By the same
argument than for $\varphi_{1}$, there is also an accepting run of
$\mathcal{A}_{\Phi}$ on $\theta^{m'}$ taking transition
$x.\delta(\varphi_{2},\sigma_{m})$ which enables to completely
construct $\rho$.
\end{proof}

\begin{proposition}
\label{PropTilde}
Let $K$ be a set of index, $\forall k \in K, \Phi_{k} := \varphi_{1,
  k} \tilde{U}_{I_{k}} \varphi_{2, k}$ be MITL formulas,
$\ell_{\Phi_{k}}$ the associated locations in an \ATA $\mathcal{A}$ of
OW, representing an MITL formula, and $J_{k} \in \mathcal{I}(\R^{+})$.
Let $\theta = (\bar{\sigma},\bar{\tau})$ be a timed word. The
automaton $\mathcal{A}$ Id-accepts $\theta$ from configuration
$\lbrace (\ell_{\Phi_{k}},J_{k})_{k \in K} \rbrace \text{ iff }
\forall k \in K, \forall v \in J_{k}$, the automaton $\mathcal{A}$
accepts $\theta$ from configuration
$\lbrace(\ell_{\Phi},[v,v])\rbrace$ (i.e. : $\forall k \in K, \forall
v \in J_{k}, (\theta,1) \models \varphi_{1, k} \tilde{U}_{I_{k}-v_{k}}
\varphi_{2, k}$).
\end{proposition}

\begin{proof}
Thanks to Lemme \ref{LemProp2}, we only need to prove the following.
Let $\Phi := \varphi_{1, k} \tilde{U}_{I_{k}} \varphi_{2, k}$ be MITL formulas, $\ell_{\Phi}$ the associated location in an \ATA $\mathcal{A}$ of OW, representing an MITL formula, and $J \in \mathcal{I}(\R^{+})$.  Let $\theta = (\bar{\sigma},\bar{\tau})$ be a timed word. The automaton $\mathcal{A}$ $Id$-accepts $\theta$ from configuration $\lbrace (\ell_{\Phi},J) \rbrace \text{ iff } \forall v \in J_{k}$, the automaton $\mathcal{A}$ accepts $\theta$ from configuration $\lbrace(\ell_{\Phi},[v,v])\rbrace$ (i.e. : $\forall v \in J_{k}, (\theta,1) \models \varphi_{1, k} \tilde{U}_{I_{k}-v_{k}} \varphi_{2, k}$)..\\
($\Rightarrow$) Let $k \in K$ and $v \in J_{k}$, we must prove the automaton $\mathcal{A}$ accepts $\theta$ from configuration $\lbrace(\ell_{\Phi},v)\rbrace$.  This proof is similar to proof of Proposition \ref{inclu}, the unique difference is that the initial state $D_{0}$ is now $\lbrace(\ell_{\Phi},v)\rbrace$.\\
($\Leftarrow$) We have an accepting $Id$-run $\rho_{v}$ of $\mathcal{A}$ on $\theta$ from each configuration $\lbrace(\ell_{\Phi},v)\rbrace$, say $C_{0}^{v} = \lbrace(\ell_{\Phi},v)\rbrace \overset{t_{1}}{\rightsquigarrow} C_{1}^{v} \overset{\sigma_{1}}{\longrightarrow_{Id}} C_{2}^{v} \overset{t_{2}}{\rightsquigarrow} \dots \overset{t_{n}}{\rightsquigarrow} C_{2n-1}^{v} \overset{\sigma_{n}}{\longrightarrow_{Id}} C_{2n}^{v}$.  From the transitions taken along these runs, we can deduce the instants in which $\varphi_{1, k}$ and $\varphi_{2, k}$ are verified (See proof of Lemma \ref{LemSingleton}). We will construct an accepting $Id$-run $\rho$' of $\mathcal{A}$ on $\theta$ from configuration $\lbrace(\ell_{\Phi},J_{k})\rbrace$, say C$_{0} = \lbrace(\ell_{\Phi},J_{k})\rbrace \overset{t_{1}}{\rightsquigarrow} C_{1} \overset{\sigma_{1}}{\longrightarrow_{Id}} C_{2} \overset{t_{2}}{\rightsquigarrow} \dots \overset{t_{n}}{\rightsquigarrow} C_{2n-1} \overset{\sigma_{n}}{\longrightarrow_{Id}} C_{2n}$.  Remark that the six transitions we can take on this run from $\ell_{\Phi}$ are : "$x.\delta(\varphi_{2, k},\sigma_{i}) \wedge x.\delta(\varphi_{1, k},\sigma_{i})$", "$x.\delta(\varphi_{2, k},\sigma_{i}) \wedge \varphi_{1, k} \tilde{U}_{I_{k}} \varphi_{2, k}$", "$x.\delta(\varphi_{2, k},\sigma_{i}) \wedge x > I_{k}$", "$x \notin I_{k} \wedge x.\delta(\varphi_{1, k},\sigma_{i})$", "$x \notin I_{k} \wedge \varphi_{1, k} \tilde{U}_{I_{k}} \varphi_{2, k}$" and "$x \notin I_{k} \wedge x > I_{k}$".  So, as long as a transition containing "$\varphi_{1, k} \tilde{U}_{I_{k}} \varphi_{2, k}$" is taken, the clock present in $\ell_{\Phi}$ is not reset an the part of configurations $C_{2i}$ associated to $\ell_{\Phi}$ will be $\lbrace(\ell_{\Phi},J_{k}+\tau_{i})\rbrace$ (assuming $\tau_{0} =0$). We distinguish several cases to construct $\rho$' :
\begin{itemize}
\item \underline{if $\varphi_{2, k}$ is verified on each reading of a letter in $K :=$} $\underset{v \in J}{\bigcup}$ \underline{$I_{k}-v$ :} then $\rho$' consists of taking the transition "$x \notin I_{k} \wedge \varphi_{1, k} \tilde{U}_{I_{k}} \varphi_{2, k}$" on each reading of a letter in an instant $\tau_{i} < K$.  In such instants, the part of configuration associated to $\ell_{\Phi}$ we are in is $\lbrace(\ell_{\Phi},J_{k}+\tau_{i})\rbrace$ and we well satisfy $\forall u \in J_{k}+\tau_{i}, u \notin I_{k}$ ; else $\exists u \in J_{k}+\tau_{i}$ such that $u \in I$ and so $u-\tau_{i} \in J_{k}$ and $u - (u- \tau_{i}) = \tau_{i} \in K$, what contradict our hypothesis. On the other hand, $\rho$' consists of taking the transition "$x.\delta(\varphi_{2, k},\sigma_{i}) \wedge \varphi_{1, k} \tilde{U}_{I_{k}} \varphi_{2, k}$" on each reading of a letter in an instant in $K$.  Finally, on the first reading of a letter after $K$, say in $\tau_{j} > K$, $\rho$' consists of taking the transition "$x \notin I_{k} \wedge x > I_{k}$".  It is possible because, then, the part of configuration associated to $\ell_{\Phi}$ we are in is $\lbrace(\ell_{\Phi},J_{k}+\tau_{j})\rbrace$ and $\forall u \in J_{k}+\tau_{j}$ : $u > I_{k}$.  To prove it, suppose that $\exists u \in J_{k}+\tau_{j} : u < I_{k} \text{ or } u \in I_{k}$.  On the first hand, if $u <I$, as $u\in J_{k}+\tau_{j}$, $\exists v \in J : u = v+\tau_{j} < I_{k}$, i.e. : $\exists v \in J : \tau_{j} < I_{k}-v$, what contradicts that $\tau_{j} > K$.  On the second hand, if $u \in I_{k}$, as $u\in J_{k}+\tau_{j}$, $\exists v \in J_{k} : u = v+\tau_{j} \in I_{k}$, i.e. : $\exists v \in J_{k} : \tau_{j} \in I_{k}-v$, what contradicts that $\tau_{j} > K$.  $\rho$' is accepting thanks to the hypothesis of this case.
\item \underline{else, $\varphi_{1, k}$ is verified in a certain
    instant in $L=\lbrace u' \vert \exists u \in K : 0 \leq u' \leq u
    \rbrace$}.\break Then, there exists a smallest instant $\tau_{i} \in L$
  such that $\varphi_{1}$ is satisfied in $\tau_{i}$.  Moreover, as
  for each $v \in J_{k}, (\theta,1) \models \varphi_{1, k}
  \tilde{U}_{I_{k}-v} \varphi_{2, k}$, each instant $\tau_{j}$ with $0
  \leq j \leq i$ and $\tau_{j} \in K$ is an instant in which
  $\varphi_{2}$ must be satisfied. We must again distinguish two cases:
\begin{itemize}
\item If $\tau_{i} < K$, then $\rho$' consists of taking the
  transition "$x \notin I_{k} \wedge \varphi_{1, k} \tilde{U}_{I_{k}}
  \varphi_{2, k}$" on each reading of a letter in an instant
  $\tau_{j}$ with $0\leq j < i$ (in such instants, we well satisfy
  $\forall u \in J_{k}+\tau_{j}$, $u \notin I_{k}$ because $\tau_{j}
  \notin K$) and taking the transition "$x\notin I_{k} \wedge
  x.\delta(\varphi_{1, k},\sigma_{i})$" with the reading of
  $\sigma_{i}$ (it is possible because $\tau_{i} \notin K$).  This run
  is accepting because the transitions chosen only verify the
  satisfaction of $\varphi_{1, k}$ in an instant in which we know this
  formula is verified.
\item If $\tau_{i} \in K$, then $\rho$' consists of : taking the
  transition "$x \notin I_{k} \wedge \varphi_{1, k} \tilde{U}_{I_{k}}
  \varphi_{2, k}$" on each reading of a letter in an instant
  $\tau_{j}$ with $0\leq j < i$ and $\tau_{j} \notin K$ (in such
  instants, we well satisfy $\forall u \in J_{k}+\tau_{j}$, $u\notin
  I_{k}$ because $\tau_{j} \notin K$) ; taking the transition
  "$x.\delta(\varphi_{2, k},\sigma_{i}) \wedge \varphi_{1, k}
  \tilde{U}_{I_{k}} \varphi_{2, k}$" on each reading of a letter in an
  instant $\tau_{j}$ with $0\leq j < i$ and $\tau_{j} \in K$ (we know
  $\varphi_{2, k}$ is verified in such instants) and taking the
  transition "$x.\delta(\varphi_{2, k},\sigma_{i}) \wedge
  x.\delta(\varphi_{1, k},\sigma_{i})$" with the reading of
  $\sigma_{i}$ (it is possible because as $\tau_{i} \in K$,
  $\varphi_{2, k}$ is satisfied in this instant).  This run is
  accepting because the transitions chosen always verify the
  satisfaction of $\varphi_{1, k}$, or $\varphi_{2, k}$, in instants
  in which we know these formulas are verified.
\end{itemize}
\end{itemize}
\end{proof}

Thanks to the previous results, we can now prove Theorem
\ref{ThmBorne}. We first recall the definition of $\mer{S}$. Let
$S=\{(\ell,I_0),(\ell,I_1),\ldots,(\ell,I_m)\}$ be a set of states of
$\Aa$, all in the same location $\ell$, with, as usual $I_0 < I_1 <
\cdots < I_m$. Then, we let $\mer{S}=\{(\ell, [0,sup(I_1)]), (\ell,
I_2),\ldots, (\ell,I_m)\}$ if $I_0=[0,0]$ and $\mer{S}=S$ otherwise,
i.e., $\mer{S}$ is obtained from $S$ by \emph{grouping} $I_0$ and
$I_1$ iff $I_0=[0,0]$, otherwise $\mer{S}$ does not modify
$S$.

\medskip

\noindent \textbf{Theorem \ref{ThmBorne} :} {\it For all MITL formula
  $\Phi$, $\appfunc{\Phi}$ is a \emph{bounded approximation function}
  and $L_{\appfunc{\Phi}}(\mathcal{A}_{\Phi}) =
  L(\mathcal{A}_{\Phi})=\sem{\Phi}$.} 
\begin{proof}
  The definition of $\appfunc{\Phi}$ guarantees it is a bounded
  approximation function. The equality
  $L(\mathcal{A}_{\Phi})=\sem{\Phi}$ have already been established and
  Theorem \ref{inclu} proves the inclusion
  $L_{\appfunc{\Phi}}(\mathcal{A}_{\Phi}) \subseteq
  L(\mathcal{A}_{\Phi})$.  In the sequel, we so present a prove of the
  last needed inclusion : $L_{\appfunc{\Phi}}(\mathcal{A}_{\Phi})
  \supseteq L(\mathcal{A}_{\Phi})$.
  \medskip

  Let $\theta = (\bar{\sigma},\bar{\tau}) \in L(\mathcal{A}_{\Phi})$.
  There is an accepting $Id$-run $\rho$ of $\mathcal{A}_{\Phi}$ on
  $\theta$, say $C_{0} = \lbrace(\ell_{0},0)\rbrace
  \overset{t_{1}}{\rightsquigarrow} C_{1}
  \overset{\sigma_{1}}{\longrightarrow_{Id}} C_{2}
  \overset{t_{2}}{\rightsquigarrow} C_{3}
  \overset{\sigma_{2}}{\longrightarrow_{Id}} \dots
  \overset{t_{n}}{\rightsquigarrow} C_{2n-1}
  \overset{\sigma_{n}}{\longrightarrow_{Id}} C_{2n}$.  We must find an
  $\appfunc{\Phi}$-accepting $\appfunc{\Phi}$-run of
  $\mathcal{A}_{\Phi}$ on $\theta$ (for $k = \max \lbrace M(\Phi),
  2. \vert L \vert \rbrace$).  Our proof is divided in two parts.  In
  the first one, we will construct an accepting $f_{\theta}$-run
  $\rho$' of $\mathcal{A}_{\Phi}$ on $\theta$ (forming intervals
  following result of Proposition \ref{PropUntil}), for a certain
  approximation function $f_{\theta}$ (later, we will show that
  $f_{\theta}$ corresponds to $\appfunc{\Phi}$).  This run will be
  $D_{0} = \lbrace(\ell_{0},J)\rbrace
  \overset{t_{1}}{\rightsquigarrow} D_{1}
  \overset{\sigma_{1}}{\longrightarrow_{f_{\theta}}} D_{2}
  \overset{t_{2}}{\rightsquigarrow} D_{3}
  \overset{\sigma_{2}}{\longrightarrow_{f_{\theta}}} \dots
  \overset{t_{n}}{\rightsquigarrow} D_{2n-1}
  \overset{\sigma_{n}}{\longrightarrow_{f_{\theta}}} D_{2n}$.
  Simultaneously, we will prove that, the way we group the clock
  copies with $f_{\theta}$, each location associated with a formula
  $\varphi_{1} U_{I} \varphi_{2}$ will contain at most $4 . \lceil
  \frac{\inf(I)}{\vert I \vert} \rceil +2$ clock copies all along
  $\rho$'.  In the second step, we will deduce from the last point
  that $\rho$' is an accepting $\appfunc{\Phi}$-run.

  In the sequel, we will use the following notation.  Assume that
  $\theta = (\bar{\sigma},\bar{\tau})$, where $\bar{\sigma} =
  \sigma_{1} \sigma_{2} \dots \sigma_{n} \text{ and } \bar{\tau} =
  \tau_{1} \tau_{2} \dots \tau_{n}$.  For all $1\leq k \leq n$, we
  denote by $\theta_{k} = (\bar{\sigma}_{k},\bar{\tau}_{k})$, where
  $\bar{\sigma}_{k} = \sigma_{k} \sigma_{k+1} \dots \sigma_{n} \text{
    and } \bar{\tau}_{k} = \tau'_{1} \tau'_{2} \dots \tau'_{n-k}$ the
  timed word such that $\forall 1 \leq i \leq n-k, \tau'_{i} =
  \tau_{i+k} - \tau_{k}$.

  \underline{\textbf{Step 1 :}} We construct $\rho$' inductively,
  using $\rho$.  We reduce the number of clock copies in the
  configuration reached after each reading of a letter.  Our method is
  to group the last clock copy associated with a location $\ell_{i}$
  with the previous interval associated with this location if it is
  possible, i.e. if we still have an accepting run thanks to
  Proposition~\ref{PropUntil} (this corresponds to the $\mer{}$
  function).  In the same time, we will prove that, in each
  configuration $D = \underset{\ell_{i} \in L}{\bigcup} D(\ell_{i})$
  reached with $\rho'$, if $\ell_{i}$ corresponds to a formula
  $\varphi_{1} U_{I_{i}} \varphi_{2}$, $\nbclocks{D(\ell_{i})} \leq 4
  . \lceil \frac{\inf(I_{i})}{\vert I_{i} \vert} \rceil +2$.

\medskip

  The induction hypothesis (at step k+1) is that we have an accepting
  run on~$\theta$:
\[
D_{0} \overset{t_{1}}{\rightsquigarrow}
  D_{1} \cdots D_{2k-1} \overset{\sigma_{k}}{\longrightarrow_{Id}} D_{2k}
  \overset{t_{k+1}}{\rightsquigarrow} E_{2k+1} \cdots E_{2n-1}
  \overset{\sigma_{n}}{\longrightarrow_{Id}} E_{2n},
\]
 such that $\forall 0 \leq j \leq 2k, \forall \ell_{i} \in L$
 corresponding to a formula $\varphi_{1} U_{I_{i}} \varphi_{2}$, we
 have that $\nbclocks{D_{j}(\ell_{i})} \leq 4 . \lceil
 \frac{\inf(I_{i})}{\vert I_{i} \vert} \rceil +2$ and $D_{2k}
 \overset{t_{k+1}}{\rightsquigarrow} E_{2k+1} \cdots E_{2n-1}
 \overset{\sigma_{n}}{\longrightarrow_{Id}} E_{2n}$ is an accepting
 $Id$-run on $\theta^{k+1}$.  Thanks to this hypothesis, we will show
 how to build
\[D_{2k}
  \overset{t_{k+1}}{\rightsquigarrow} D_{2k+1}
  \overset{\sigma_{k+1}}{\longrightarrow_{Id}} D_{2k+2}
  \overset{t_{k+2}}{\rightsquigarrow} F_{2k+3} \cdots F_{2n-1}
  \overset{\sigma_{n}}{\longrightarrow_{Id}} F_{2n}\] 
such that 
\begin{enumerate}
\item[(i)] $D_{0} \overset{t_{1}}{\rightsquigarrow} D_{1} \cdots
D_{2k} \overset{\sigma_{k}}{\longrightarrow_{Id}} D_{2k+1}
\overset{\sigma_{k+1}}{\longrightarrow} D_{2k+2}
\overset{t_{k+2}}{\rightsquigarrow} F_{2k+3} \cdots F_{2n-1}
\overset{\sigma_{n}}{\longrightarrow_{Id}} F_{2n}$ is an accepting run
on $\theta$,
\item[(ii)] $\forall \ell_{i} \in L$ corresponding to a formula
$\varphi_{1} U_{I_{i}} \varphi_{2}$ : $\nbclocks{D_{2k+2}(\ell_{i})}
\leq 4 . \lceil \frac{\inf(I_{i})}{\vert I_{i} \vert} \rceil +2$,
\item[(iii)] $D_{2k+2} \overset{t_{k+2}}{\rightsquigarrow} F_{2k+3} \cdots
F_{2n-1} \overset{\sigma_{n}}{\longrightarrow_{Id}} F_{2n}$ is an
accepting $Id$-run on $\theta^{k+1}$.
\end{enumerate}

  \underline{Basis :} (k=0) we define $D_{0}
  = C_{0} = \lbrace(\ell_{0},[0,0])\rbrace$.  We still have an
  accepting $Id$-run of $\mathcal{A}_{\Phi}$ on $\theta$ : $\rho$.
  (The number of copies in $\ell_{0}$ will be discussed later because
  this location does not always correspond to a formula $\varphi_{1}
  U_{I_{i}} \varphi_{2}$.) 

  \underline{Induction :} (k+1)\\
  We know there is an accepting $Id$-run of $\mathcal{A}_{\Phi}$ from
  $D_{2k}$ on $\theta^{k+1}$, say the two first steps of this $Id$-run
  are : $D_{2k} \overset{t_{k+1}}{\rightsquigarrow} E_{2k+1}
  \overset{\sigma_{k+1}}{\longrightarrow_{Id}} E_{2(k+1)}$.  We define
  configuration $D_{2k+1}$ of $\rho$' as $D_{2k+1} := E_{2k+1}$. As
  $\forall \ell_{i} \in L$ corresponding to a formula $\varphi_{1}
  U_{I_{i}} \varphi_{2}$, $\nbclocks{D_{2k}(\ell_{i})} \leq 4 . \lceil
  \frac{\inf(I_{i})}{\vert I_{i} \vert} \rceil +2$ (induction
  hypothesis), we also have $\nbclocks{D_{2k+1}(\ell_{i})} \leq 4
  . \lceil \frac{\inf(I_{i})}{\vert I_{i} \vert} \rceil +2$.  Assume
  that $E_{2k+2}(\ell_{i}) = \lbrace J^{i}_{1}, ..., J^{i}_{m_{i}}
  \rbrace$. We define $D_{2k+2}$ of $\rho$' as $f_{\theta}(E_{2k+2}) =
  \underset{\ell_{i} \in L}{\bigcup} f_{\theta}(E_{2k+2}(\ell_{i}))$,
  where $\forall \ell_{i} \in L, f_{\theta}(E_{2k+2}(\ell_{i}))$ is
  defined as follows.
  \begin{align*}
    f_{\theta}(E_{2k+2}(\ell_{i})) &=
    \left\{
      \begin{array}{ll}
        \mer{E_{2k+2}(\ell_{i})} &\text{if } \exists m \geq 1 \text{ such that : } \; (\theta^{k+2},m) \vDash \varphi_{2}\\
        &\phantom{if }\wedge \tau'_{m} \in I_{i} - sup(J_{2}^{i}) \wedge \tau'_{m} \in I_{i}\\
        &\phantom{if }\wedge ( \forall 1 \leq m' < m : (\theta^{k+2},m') \models \varphi_{1} )\\
        E_{2k+2}(\ell_{i}) &\text{otherwise}
      \end{array}
    \right.
  \end{align*}
  Where, as defined above, 
  $$
  \mer{E_{2k+2}(\ell_{i})} =
  \lbrace(\ell_{i}, [0, \sup(J_{2}^{i})] ), (\ell_{i},J_{3}^{i}),
  (\ell_{i},J_{4}^{i}), \dots, (\ell_{i},J_{m_{i}}^{i}) \rbrace
  $$

We must prove there is an accepting $Id$-run of
$\mathcal{A}_{\Phi}$ from $D_{2.(k+1)} = f_{\theta}(E_{2k+2}) =
\underset{\ell_{i} \in L}{\bigcup} f(E_{2k+2}(\ell_{i}))$ on
$\theta^{k+2}$.  Let $\ell_{i} \in L$, thanks to Proposition~\ref{PropUntil}, it is sufficient to prove that there is an
accepting $Id$-run of $\mathcal{A}_{\Phi}$ on $\theta^{k+2}$ from
$D_{2.(k+1)}(\ell_{i}) := f_{\theta}(E_{2k+2}(\ell_{i}))$. If
$f_{\theta}(E_{2k+2}(\ell_{i})) = E_{2k+2}(\ell_{i})$, the
accepting $Id$-run given by induction hypothesis on
$E_{2k+2}(\ell_{i})$ can always be used. Else,
$$
f_{\theta}(E_{2k+2}(\ell_{i})) = \lbrace (\ell_{i}, [0,
\sup(J_{2}^{i})] ), (\ell_{i},J_{3}^{i}), (\ell_{i},J_{4}^{i}), \dots,
(\ell_{i},J_{m_{i}}^{i}) \rbrace
$$ 
As in this case $E_{2k+2}(\ell_{i})$ was $\lbrace(\ell_{i},[0,0]),
(\ell_{i},J_{2}^{i}), (\ell_{i},J_{3}^{i}), (\ell_{i},J_{4}^{i}),
\dots,\break (\ell_{i},J_{m_{i}}^{i})\rbrace$, the accepting $Id$-run
given by induction hypothesis can be used from $\lbrace
(\ell_{i},J_{3}^{i}),$ $(\ell_{i},J_{4}^{i}), \dots,
(\ell_{i},J_{m_{i}}^{i})\rbrace \subseteq
f_{\theta}(E_{2k+2}(\ell_{i}))$ (thanks to Proposition
\ref{PropUntil}) and we only need to prove there is an accepting
$Id$-run of $\mathcal{A}_{\Phi}$ on $\theta^{k+2}$ from\break
$\lbrace(\ell_{i},[0, \sup(J_{2}^{i})]\rbrace$.

As in this case $\exists m \geq 1 \text{ such that : } \;
(\theta^{k+2},m) \vDash \varphi_{2} \wedge \tau'_{m} \in I_{i} -
sup(J_{2}^{i}) \wedge \tau'_{m} \in I_{i} \wedge ( \forall 1 \leq m' <
m : (\theta^{k+2},m') \models \varphi_{1} )$, we can conclude thanks
to Proposition \ref{PropUntil}.
\medskip

We must now show that, the way we grouped clock copies with
$f_{\theta}$, $\forall \ell_{i} \in L$ corresponding to a formula
$\varphi_{1} U_{I_{i}} \varphi_{2}$, $\nbclocks{D_{2k+2}(\ell_{i})}
\leq 4 . \lceil \frac{\inf(I_{i})}{\vert I_{i} \vert} \rceil +2$.  We
prove it by contradiction.

Let us suppose that $\nbclocks{D_{2k+2}(\ell_{i})} > 4 . \lceil
\frac{\inf(I_{i})}{\vert I_{i} \vert} \rceil +2$, for a certain
location $\ell_{i}$ corresponding to formula $\varphi_{1} U_{I_{i}}
\varphi_{2}$.  We so have more than $2.\lceil \frac{\inf(I_{i})}{\vert
  I_{i} \vert} \rceil +1$ intervals associated with $\ell_{i}$ in
$D_{2k +2}$, i.e : $D_{2k+2}(\ell_{i}) = \lbrace(\ell_{i},J_{1}^{i}),
(\ell_{i},J_{2}^{i}), \dots, (\ell_{i},J_{m_{i}}^{i})\rbrace$, for a
certain $m_{i} > 2.\lceil \frac{\inf(I_{i})}{\vert I_{i} \vert} \rceil
+1$.  The way we grouped clock copies with $f_{\theta}$, we know that
each interval $J_{m_{i}}^{j}$, for $1 \leq j \leq m_{i}$, satisfies
the following property :
\begin{eqnarray}
\label{eq1}
\exists k _{j} \geq 1 \text{ such that : } (\theta^{k+1}, k_{j})
\vDash \varphi_{2} \; \wedge \; \tau'_{k_{j}} \in I_{i} -
\sup(J^{i}_{j}) \nonumber\\ \wedge \; \tau'_{k_{j}} \in I_{i} -
\inf(J^{i}_{j}) \; \wedge \; ( \forall 1 \leq k' < k_{j},
(\theta^{k+1},k') \vDash \varphi_{1} ).
\end{eqnarray}
(If it is not the case anymore, $\rho$' would not be an accepting
run.) Moreover, $\forall 1 < j \leq m_{i}$, we have the following
property 
:
\begin{eqnarray} 
  \label{eq2}
  \forall k_{m} \geq 1 : (\theta^{k+1}, k_{m}) \nvDash \varphi_{2} \; \vee \; \tau'_{k_{m}} \notin I_{i} - \sup(J_{j}^{i}) \; \vee \; \tau'_{k_{m}} \notin I_{i} - \sup(J_{j-1}^{i}) \nonumber\\
  \vee \; \exists 1 \leq k' < k_{m}, (\theta^{k+1},k') \nvDash \varphi_{1}.
\end{eqnarray}
(If it is not the case, $\sup(J_{j-1}^{i})$ would have been grouped
with $J_{j}^{i}$.)

We first claim that $\forall 3 \leq j \leq m_{i}$,
$\sup(J_{j}^{i})-\sup(J_{j-2}^{i}) \geq \vert I_{i} \vert$.  We will
prove it by contradiction.  Let $j^\star$ be s.t. $3 \leq j^\star \leq
m_{i}$ and suppose that $\sup(J_{j^\star}^{i})-\sup(J_{j^\star-2}^{i})
< \vert I_{i} \vert$, i.e : $\sup(J_{j^\star}^{i}) < \vert I_{i} \vert
+\sup(J_{j^\star-2}^{i})$.  Then $(I_{i}-\sup(J_{j^\star}^{i})) \cap
(I_{i}-\sup(J_{j^\star-2}^{i})) \neq \emptyset$ because these
intervals have the same size, moreover, as $\sup(J_{j^\star}^{i}) >
\sup(J_{j^\star-2}^{i})$, $\inf(I_{i}-\sup(J_{j^\star}^{i}))) <
\inf((I_{i}-\sup(J_{j^\star-2}^{i}))$ and finally :
$\sup(I_{i}-\sup(J_{j^\star}^{i})) = \sup(I_{i})-\sup(J_{j^\star}^{i})
> \sup(I_{i}) -(\vert I_{i} \vert +\sup(J_{j^\star-2}^{i})) =
\sup(I_{i}) -\sup(I_{i}) +\inf(I_{i}) -\sup(J_{j^\star-2}^{i}) =
\inf(I_{i})-\sup(J_{j^\star-2}^{i}) =
\inf(I_{i}-\sup(J_{j^\star-2}^{i}))$.  So, as $\sup(J_{j^\star-2}^{i})
< \sup(J_{j^\star-1}^{i}) < \sup(J_{j^\star}^{i})$,
$I_{i}-\sup(J_{j^\star-1}^{i}) \subseteq (I_{i}-\sup(J_{j^\star}^{i}))
\cup (I_{i}-\sup(J_{j^\star-2}^{i}))$.  Equation~$\eqref{eq1}$,
letting $j=j^\star$ implies that $\tau'_{k_{j^\star-1}} \in
I_{i}-\sup(J_{j^\star-1}^{i})$, so $\tau'_{k_{j^\star-1}} \in
(I_{i}-\sup(J_{j^\star}^{i})) \cup (I_{i}-\sup(J_{j^\star-2}^{i}))$.
Though, if $\tau'_{k_{j^\star-1}} \in (I_{i}-\sup(J_{j^\star}^{i}))$,
we contradict $\eqref{eq2}$ in $j^\star$ taking $k_{m} =
k_{j^\star-1}$ (thanks to the definition of $k_{j^\star-1}$) ; and if
$\tau'_{k_{j^\star-1}} \in (I_{i}-\sup(J_{j^\star-2}^{i}))$, we
contradict $\eqref{eq2}$ in for $j=j^\star-1$ taking $k_{m} =
k_{j^\star-1}$ (thanks to the definition of $k_{j^\star-1}$).

We now know that $\forall 3 \leq j \leq m_{i}$,
$\sup(J_{j}^{i})-\sup(J_{j-2}^{i}) \geq \vert I_{i} \vert$. So,
$\sup(J_{m_{i}}^{i})-\sup(J_{1}^{i}) \geq \lceil \frac{(m_{i}-2)}{2}
\rceil . \vert I_{i} \vert$.  As $m_{i} > 2.\lceil
\frac{\sup(I_{i})}{\vert I_{i} \vert} \rceil +1$, we have that:
\begin{align*}
  \sup(J_{m_{i}}^{i})-\sup(J_{1}^{i}) & > \left\lceil
    \frac{(2.\left\lceil \frac{\sup(I_{i})}{\vert I_{i} \vert}
      \right\rceil +1 -2)}{2} \right\rceil . \vert I_{i}\vert\\ & =
  \left\lceil \left\lceil \frac{\sup(I_{i})}{\vert I_{i} \vert}
    \right\rceil - \frac{1}{2} \right\rceil . \vert I_{i}\vert =
  \left\lceil \frac{\sup(I_{i})}{\vert I_{i} \vert} \right\rceil
  . \vert I_{i}\vert \geq \sup(I_{i}).
\end{align*}
It means that $\sup(J_{m_{i}}^{i})-\sup(J_{1}^{i}) > \sup(I_{i})$, and
hence $\sup(J_{m_{i}}^{i}) > \sup(I_{i})$.  It is a contradiction
because if $\sup(J_{m_{i}}^{i}) > \sup(I_{i})$, we can not have an
accepting $Id$-run from
$\lbrace(\ell_{i},\sup(J^{i}_{m_{i}}))\rbrace$ and therefore neither
from $D_{2k +2}$, while we have just proved it is the case.

So far, we have showed that $\forall \ell \in L, \forall 1 \leq j \leq
n : \nbclocks{D_{j}(\ell_{i})} \leq 4 . \lceil
\frac{\inf(I_{i})}{\vert I_{i} \vert} \rceil +2$.  So, the bound we
have on the run is $\vert L \vert . 4 . \lceil
\frac{\inf(I_{i})}{\vert I_{i} \vert} \rceil +2$.  In Step 2, we will
show we can improve this bound by $M(\Phi)$ and so conclude that our
$f_{\theta}$-run is in fact an $\appfunc{\Phi}$-run.

\textbf{Case $\varphi_{1} \tilde{U}_{I_{i}} \varphi_{2}$ :} The
arguments are similar to the $U$ case, using
Proposition~\ref{PropTilde}. The bound found is the following : $\forall \ell \in L$ associated to a formula of the form $\varphi_{1} \tilde{U}_{I_{i}} \varphi_{2}$, $\forall 1 \leq j \leq
n : \nbclocks{D_{j}(\ell_{i})} \leq 2 . \lceil
\frac{\sup(I_{i})}{\vert I_{i} \vert} \rceil +2$.
Remark that, to prove this case, we must
assume the following property on $\rho$: for all $0\leq i\leq 2n$, for
all $\ell$ corresponding to a sub-formula of the form
$\varphi_1\tilde{U}_i\varphi_2$, for all $J\in C_i(\ell)$: $inf(J)\leq
sup(I)$. Remark that this is always possible, because if an interval
$J$ is present in a location $\ell$ corresponding to
$\varphi_1\tilde{U}_i\varphi_2$, with $inf(J)>sup(I)$, the arc
$(\ell,\sigma, x\not\in I\wedge x>sup(I))$ can be taken for all
$\sigma\in \Sigma$.

\underline{\textbf{Step 2 :}}\\ We still must prove that $\rho$' is an
accepting $\appfunc{\Phi}$-run : thanks to Step 1, it remains to prove
that each configuration reached by $\rho$' contains at most $M(\Phi)$
clock copies.  By definition of the transitions starting from the
initial location $\varphi_{init,\Phi}$, at most two clock copies will
be associated with this location (because the initial state is
$\lbrace(\ell_{0},[0,0])\rbrace$) and it will have no clock copy
associated with this location anymore as soon as clock copies are sent
towards other locations.  Moreover, all other locations of
$\mathcal{A}_{\Phi}$ are locations associated with
sub-formulas  of $\Phi$ of type
$\varphi_{1}U_{I_{i}} \varphi_{2}$ : we know such a location contains
at most $4.\lceil \frac{\inf(I_{i})}{\vert I_{i} \vert} \rceil +2$
clock copies all along $\rho$'.  Remark that the transition starting
from the location of a formula $\varphi_{1} U_{I} \varphi_{2}$ is
$(x.\delta(\varphi_{2},\sigma) \wedge x \in I) \vee
(x.\delta(\varphi_{1},\sigma) \wedge \varphi_{1}U_{I} \varphi_{2}
\wedge x \leq \sup(I))$ : it means that $\delta(\varphi_{1},\sigma)$
is taken a lot of times while $\delta(\varphi_{2},\sigma)$ is only
taken once.
It is why we distinguish in the definition of $M(\Phi)$ the maximal
number of copies present in configurations reached by $\rho'$ : (1) to
verify a sub-formula $\varphi$ of $\Phi$ that receives a lot of clock
copies (2) to verify a sub-formula $\varphi$ of $\Phi$ that receives
at most one clock copy (3) to verify $\varphi = \Phi$, with the
complete automaton $\mathcal{A}_{\varphi}$.

It is not difficult to be convinced that a proof by induction on the
structure of $\Phi$ enables to show that each configuration of
$\mathcal{A}_{\Phi}$ reached by $\rho$' contains at most $M(\Phi)$
clock copies.

\end{proof}

\section{Towards a timed automaton\label{sec:towards-timed-autom}}
Let $\Phi$ be an MITL formula, and assume $\Aa_\Phi=(\Sigma, L^\Phi,
\ell_0^\Phi, F^\Phi,\delta^\Phi)$. Let us show how to build the TA
$\Bb_\Phi=(\Sigma, L,\ell_0, X, F, \delta)$
s.t. $L(\Bb_\Phi)=L_{\appfunc{\Phi}}(\Aa_\Phi)$. The components of $\Bb_\Phi$
are as follows. For a set of clocks $X$, we let $\loc(X)$ be the set
of functions $S$ that associate to each $\ell\in L^\Phi$ a finite
sequence $(x_1,y_1),\ldots,(x_n,y_n)$ of pairs of clocks from $X$,
s.t. each clock occurs only once in all the $S(\ell)$. Then:
\begin{itemize}
\item $L=\loc(X)$. Intuitively, a configuration $(S, v)$ of $\Bb_\Phi$
  encodes the configuration $C$ of $\Aa_\Phi$ s.t. for all $\ell\in
  L^\Phi$: $C(\ell)=\{[v(x), v(y)]\mid (x,y)\in S(\ell)\}$.
\item $\ell_0$ is s.t. $\ell_0(\ell_0^\Phi)=(x,y)$, where $x$ and $y$
  are two clocks arbitrarily chosen from $X$, and
  $\ell_0(\ell)=\emptyset$ for all $\ell\in L^\Phi\setminus
  \{\ell_0^\Phi\}$.
\item $X$ is a set of clocks s.t. $|X|=M(\Phi)$.
\item $F$ is the set of all locations $S$ s.t. $\{\ell\mid
  S(\ell)\neq\emptyset\}\subseteq F^\Phi$.
\end{itemize}
Finally, we must define the set of transitions $\delta$ to let
$\Bb_\Phi$ simulate the executions of $\Aa_\Phi$. First, we observe
that, for each location $\ell\in L^\Phi$, for each $\sigma\in\Sigma$,
all \emph{arcs} in $\delta^\Phi$ are either of the form
$(\ell,\sigma,\mathit{true})$ or $(\ell,\sigma,\mathit{false})$ or of
the form $\big(\ell,\sigma,\ell\wedge
x.(\ell_1\wedge\cdots\wedge\ell_k)\wedge g\big)$ or of the form
$\big(\ell,\sigma,x.(\ell_1\wedge\cdots\wedge\ell_k)\wedge g\big)$,
where $g$ is \emph{guard} on $x$, i.e. a finite conjunction of clock
constraints on $x$. Let $S\in\configs{\Bb_\Phi}$ be a configuration of
$\Bb_\Phi$, let $\ell\in L^\Phi$, let $\sigma\in\Sigma$ be a
letter. Let $(x,y)$ be a pair of clocks occurring in $S(\ell)$ and let
us associate to this pair an arc $a$ of $\delta^\Phi$ of the form
$(\ell,\sigma,\gamma)$. Then, we associate to $a$ a \emph{guard}
$\gu{a}$, and two sets $\re{a}$ and $\de{a}$, defined as follows:
\begin{itemize}
\item if $\gamma\in\{\mathit{true},\mathit{false}\}$, then, $\gu{a}=a$
  and $\re{a}=\de{a}=\emptyset$.
\item if $\gamma$ is of the form
  $x.(\ell_1\wedge\cdots\wedge\ell_k)\wedge g$, then $\gu{a}=g$,
  $\re{a}=\{\ell_1,\ldots,\ell_k\}$ and $\de{a}=\emptyset$.
\item  if $\gamma$ is of the form
  $\ell\wedge x.(\ell_1\wedge\cdots\wedge\ell_k)\wedge g$, then $\gu{a}=g$,
  $\re{a}=\{\ell_1,\ldots,\ell_k\}$ and $\de{a}=\{(x,y)\}$.
\end{itemize}
Thanks to those definitions, we can now define $\delta$. Let $S$ be a
location in $L$, and assume:
\begin{align*}
  \{(\ell_1,x_1,y_1),\ldots, (\ell_k,x_k, y_k)\} &= \{(\ell, x,y)\mid
  (x,y)\in S(\ell)\}
\end{align*}
Then $(S, \sigma, g, r, S')\in \delta$ iff there are: a set
$A=\{a_1,\ldots, a_k\}$ of arcs s.t.:
\begin{itemize}
\item for all $1\leq i\leq k$: $a_i$ is an arc of $\delta^\Phi$ of the
  form $(\ell,\sigma,\gamma_k)$, associated to $(x_i, y_i)$.
\item For each $\ell\in L^\Phi$, we let
  $\bar{S}_\ell=(x_1',y_1')(x_2',y_2')\cdots(x_m',y_m')$ be obtained
  from $S(\ell)$ by deleting all pairs $(x,y)\not\in\bigcup_{1\leq
    i\leq k}\de{a_i}$. Then, for all $\ell\in L^\Phi$:
  \begin{align*}
    S'(\ell) &\in \big\{(x,y)\cdot \bar{S}_\ell,
    (x,y_1')(x_2',y_2')\cdots(x_m',y_m')\big\} &\text{if }\ell\in\bigcup_{1\leq i\leq k}\re{a_i}\\
    S'(\ell) &=\bar{S}_\ell &\text{otherwise}
  \end{align*}
  When $S'(\ell)=(x,y)\cdot \bar{S}_\ell$, we let $R_\ell=\{x,y\}$;
  when $S'(\ell)=(x,y_1')(x_2',y_2')\cdots\break(x_m',y_m')$, we let
  $R_\ell=\{x\}$; and we let $R_\ell=\emptyset$ otherwise.
\item $g=\bigwedge_{1\leq i\leq k}(\gu{a_i}[x/x_i]\wedge
  \gu{a_i}[x/y_i])$.
\item $r=\cup_{\ell\in L^\Phi}R_\ell$.

\end{itemize}

For all MITL formula $\Phi$, let $\Ii_\Phi$ be the set of all the
intervals that occur in $\Phi$. Then:\medskip

\noindent{\bf Theorem \ref{theo:TA-from-MITL}}:~{\it For all MITL
  formula $\Phi$, $\Bb_\Phi$ has $M(\Phi)$ clocks and $O( (\vert
  \Phi \vert)^{(m . \vert \Phi \vert)})$ locations, where $m =
  \max_{I\in\Ii_\Phi} \left\{ 2 \times \left\lceil
      \frac{\inf(I)}{\vert I \vert} \right\rceil + 1, \left\lceil
      \frac{\sup(I)}{\vert I \vert} \right\rceil + 1 \right\}$.  }
\begin{proof}
  By definition of $\mathcal{B}_\Phi$, $\vert X \vert = M(\Phi) = O( 2
  . m . \vert \Phi \vert )$.  Moreover, one location of this automaton
  is an association, to each location $\ell$ of $\mathcal{A}_{\Phi}$,
  of a finite sequence $(x_1,y_1), \dots, (x_n,y_n)$ of pairs of
  clocks from $X$ such that each pair is associated to a unique
  $\ell$.  In other words, each couple of clocks $(x_i,y_i)$ can be
  associated to : either one and only one of the $\ell \in L$ or to no
  $\ell \in L$.  We so have $\vert L \vert + 1$ possibilities of
  association of each pair $(x_i,y_i)$ and we have $\frac{M(\Phi)}{2}$
  such pairs.  So, $\mathcal{B}_\Phi$ has $(\vert L \vert +
  1)^{\frac{M(\Phi)}{2}}$ locations, i.e. : $O\left( (\vert \Phi
    \vert)^{m . \vert \Phi \vert} \right) = O\left( 2^{m . \vert \Phi
      \vert . log_{2}(\vert \Phi \vert)} \right)$ (because $\vert L
  \vert = O(\vert \Phi \vert)$ and $M(\Phi) = O( 2 . m . \vert \Phi
  \vert )$).

  We prove that $\Bb_\Phi$ recognizes $\sem{\Phi}$ by mapping each
  configuration of $\mathcal{B}_\Phi$ to a configuration of
  $\mathcal{A}_{\Phi}$ and conversely and that this mapping is
  consistent with all runs.

  First, let $(S,v)$ be a configuration of $\mathcal{B}_\Phi$, we map
  it to the following configuration of $\mathcal{A}_{\Phi}$. We know
  that $\forall \ell \in L$, $S(\ell)$ is a finite sequence
  $(x_1,y_1), \dots, (x_n,y_n)$ of pairs of clocks from : it
  corresponds to the (unique) configuration of $\mathcal{A}_{\Phi}$,
  $C = \bigcup_{\ell \in L} C(\ell)$ where $C(\ell) = \lbrace
  [v(x),v(y)] \vert (x,y) \in S(\ell) \rbrace$.  It is straightforward
  to see that, if $(S,v) \timestep{t} (S',v')$ and $(S,v)$ is mapped
  to $C$, there exists $C'$ such that $C \timestep{t} C'$ and $C'$ is
  mapped to $(S,v)$.  Moreover, we claim that, if $(S,v)
  \xrightarrow{\sigma} (S',v')$and $(S,v)$ is mapped to $C$, there
  exists $C'$ such that $C \xrightarrow{\sigma}_{\appfunc{\Phi}} C'$
  and $C'$ is mapped to $(S,v)$.  This holds because, if $(S,v)
  \xrightarrow{\sigma} (S',v')$, we can use the arc $a \in
  \delta^{\Phi}(\ell,\sigma)$ associated $(x,y) \in S(\ell)$ to find a
  minimal model of the state $(\ell,[v(x),v(y)])$ of $C$, this way, we
  reach a configuration $C'$ of $\mathcal{A}_{\Phi}$ that is mapped to
  $(S',v')$ thanks to the definition of $\delta$ : corresponding
  clocks are reset in the same time ; we verify the same guards on
  corresponding clocks ; the configuration we can reach in
  $\mathcal{B}_\Phi$ corresponds, for each location $\ell \in L$ whose
  smallest associated interval is [0,0], to group or not this interval
  with the second associated with $\ell$, what correspond to the
  configurations of $\mathcal{A}_{\Phi}$ we can reach from $C$.

  Second, let $C$ be a configuration of $\mathcal{A}_{\Phi}$, we map
  it to the set of all $(S,v)$ s.t. for all $\ell\in L$: $C(\ell)=
  \lbrace I_1^\ell, I_2^\ell,\ldots, I_n \rbrace$ iff $v(x_1) =
  \inf(I_1)$, $v(y_1) = \sup(I_1)$,\ldots, $v(x_n) = \inf(I_n)$,
  $v(y_n) = \sup(I_n)$. Observe that there are indeed several
  configurations $(S,v)$ of $\Bb_\Phi$ that satisfy this definition:
  they can all be obtained up to clock renaming.  To keep a
  consistence in our runs, we must only choose the corresponding
  configuration of $\mathcal{B}_\Phi$ such that: once a pair of clocks
  is associated to an interval $I_j$ of $C(\ell)$, if $I_j$ is still
  in $C'(\ell)$, the same clocks represents its bounds.  In the same
  way, when an interval $I'_j$ of the form $[0,\sup(I_j)]$ is in
  $C'(\ell)$, the same clocks represents its bounds.  In contrary,
  when a new interval $I_j ( = [0,0] )$ is associated to $C(\ell)$, we
  can arbitrary choose which unused pair of clocks $(x_i,y_i)$ will
  represent it.  Thanks to this trick, we can proof properties similar
  those of the first step.
\end{proof}


\end{document}